\documentclass{amsart}
\usepackage{times,amssymb,amsmath,amsfonts,float,nicefrac,color}
\usepackage{euscript,graphics}
\usepackage[dvips]{epsfig}

\newtheorem{theorem}{Theorem}[section]

\newtheorem{lemma}[theorem]{Lemma}
\newtheorem{proposition}[theorem]{Proposition}
\newtheorem{corollary}[theorem]{Corollary}
\newtheorem{definition}{Definition}

\renewcommand\qed{\rule{1mm}{2mm}\medskip}

\newcommand{\remove}[1]{}
\renewcommand{\tilde}{\widetilde}

\newcommand\wt{\mbox{{\rm wt}$_H$}}

\def\supp{\qopname\relax{no}{supp}}
\def\avg{{\text{\sf E}}}
\def\sgn{\qopname\relax{no}{sgn}}
\def\rank{\qopname\relax{no}{rk}}

\newcommand\nc\newcommand
\nc\bfa{{\boldsymbol a}}\nc\bfA{{\bf A}}\nc\cA{{\mathcal A}}
\nc\bfb{{\boldsymbol b}}\nc\bfB{{\bf B}}\nc\cB{{\mathcal B}}
\nc\bfc{{\boldsymbol c}}\nc\bfC{{\bf C}}\nc\cC{{\mathcal C}}
\nc\bfd{{\boldsymbol d}}\nc\bfD{{\bf D}}\nc\cD{{\mathcal D}}
\nc\bfe{{\boldsymbol e}}\nc\bfE{{\bf E}}\nc\cE{{\mathcal E}}
\nc\bff{{\boldsymbol f}}\nc\bfF{{\bf F}}\nc\cF{{\mathcal F}}
\nc\bfg{{\boldsymbol g}}\nc\bfG{{\bf G}}\nc\cG{{\mathcal G}}
\nc\bfh{{\boldsymbol h}}\nc\bfH{{\bf H}}\nc\cH{{\mathcal H}}
\nc\bfi{{\boldsymbol i}}\nc\bfI{{\bf I}}\nc\cI{{\mathcal I}}
\nc\bfj{{\boldsymbol j}}\nc\bfJ{{\bf J}}\nc\cJ{{\mathcal J}}
\nc\bfk{{\boldsymbol k}}\nc\bfK{{\bf K}}\nc\cK{{\mathcal K}}
\nc\bfl{{\boldsymbol l}}\nc\bfL{{\bf L}}\nc\cL{{\mathcal L}}
\nc\bfm{{\boldsymbol m}}\nc\bfM{{\bf M}}\nc\cM{{\mathcal M}}
\nc\bfn{{\boldsymbol n}}\nc\bfN{{\bf N}}\nc\cN{{\mathcal N}}
\nc\bfo{{\boldsymbol o}}\nc\bfO{{\bf O}}\nc\cO{{\mathcal O}}
\nc\bfp{{\boldsymbol p}}\nc\bfP{{\bf P}}\nc\cP{{\EuScript P}}
\nc\bfq{{\boldsymbol q}}\nc\bfQ{{\bf Q}}\nc\cQ{{\mathcal Q}}
\nc\bfr{{\boldsymbol r}}\nc\bfR{{\bf R}}\nc\cR{{\mathcal R}}
\nc\bfs{{\boldsymbol s}}\nc\bfS{{\bf S}}\nc\cS{{\mathcal S}}
\nc\bft{{\boldsymbol t}}\nc\bfT{{\bf T}}\nc\cT{{\mathcal T}}
\nc\bfu{{\boldsymbol u}}\nc\bfU{{\bf U}}\nc\cU{{\mathcal U}}
\nc\bfv{{\boldsymbol v}}\nc\bfV{{\bf V}}\nc\cV{{\mathcal V}}
\nc\bfw{{\boldsymbol w}}\nc\bfW{{\bf W}}\nc\cW{{\mathcal W}}
\nc\bfx{{\boldsymbol x}}\nc\bfX{{\bf X}}\nc\cX{{\mathcal X}}
\nc\bfy{{\boldsymbol y}}\nc\bfY{{\bf Y}}\nc\cY{{\mathcal Y}}
\nc\bfz{{\boldsymbol z}}\nc\bfZ{{\bf Z}}\nc\cZ{{\mathcal Z}}
\nc\od{{\bar d}}\nc\ow{{\bar w}}\nc\odelta{{\bar\delta}}
\nc\ox{{\bar x}}\nc\oy{{\bar y}}\nc\ou{{\bar u}}
\nc\oh{{\bar h}}

\newcommand\reals{{\mathbb R}}

\newcommand\ff{{\mathbb F}}

\nc\ellone{{\ell_1}}
\nc\elltwo{{\ell_2}}
\nc\ellinf{{{\ell_\infty}}}
\nc\ip[2]{\langle #1,#2\rangle}

\newcommand{\beeq}{\begin{eqnarray*}}
\newcommand{\eneq}{\end{eqnarray*}}

\newcommand{\half}{\nicefrac12}

\begin{document}
\title[Random subdictionaries and sparse signal recovery]	
{Random subdictionaries and coherence conditions for sparse signal recovery}
\thanks{{\em Date}\/: \today.\/ }%

\author[A. Barg]{Alexander Barg$^\ast$}\thanks{$^\ast$
Dept. of Electrical and Computer Engineering and Institute for Systems
Research, University of Maryland, College Park, MD 20742,
and Institute for Problems of Information Transmission,
Russian Academy of Sciences, Moscow, Russia. Email: abarg@umd.edu.
Research supported in part by NSF grants
CCF0916919, CCF1217245, CCF1217894, DMS1101697, and NSA H98230-12-1-0260.}\author[A. Mazumdar]{Arya Mazumdar$^\dag$}\thanks{$^\dag$
Research Laboratory of Electronics, Massachusetts Institute of
Technology, Cambridge, MA 02139. This work was done while the
author was at the University of Maryland, College Park, MD. Email: aryam@mit.edu}
\author[R. Wang]{Rongrong Wang$^\ddag$}\thanks{$^\ddag$
Dept. of Mathematics, University of Maryland, College Park, MD 20742. Email: rwang928@math.umd.edu}

\begin{abstract}
The most frequently used
condition for sampling matrices employed in compressive sampling
is the restricted isometry (RIP) property of the matrix when restricted to sparse
signals. At the same time, imposing this condition makes it difficult to find explicit
matrices that support recovery of signals from sketches of the
optimal (smallest possible) dimension. A number of attempts
have been made to relax or replace the RIP property in sparse recovery
algorithms. We focus on the relaxation under which the near-isometry property
holds for most rather than for all submatrices of the sampling matrix, known
as statistical RIP or StRIP condition. 
We show that sampling matrices of dimensions $m\times N$ with maximum 
coherence $\mu=O((k\log^3 N)^{-1/4})$ and mean square coherence $\bar\mu^2=O(1/(k\log N))$ 
support stable recovery of $k$-sparse signals using Basis Pursuit.
These assumptions are satisfied in many examples. As a result, we
are able to construct sampling matrices that support recovery with low error
for sparsity $k$ higher than $\sqrt m,$ which exceeds the range of parameters
of the known classes of RIP matrices.
\end{abstract}
\maketitle

\section{Introduction}
One of the important problems in theory of compressed sampling is construction
of sampling operators that support algorithmic procedures of sparse recovery.
A universal sufficient condition for stable reconstruction is given by
the restricted isometry property (RIP) of sampling matrices \cite{can08a}. It has been shown 
that sparse signals compressed to low-dimensional images using linear RIP 
maps can be reconstructed using $\ell_1$ minimization procedures such as Basis pursuit
and Lasso \cite{can05,can06b,can08a,cai09},

Let $\bfx$ be an $N$-dimensional real signal that has
a sparse representation in a suitably chosen basis.
We will assume that $\bfx$ has $k$ nonzero coordinates (it is a {\em $k$-sparse} vector) or is
approximately sparse in the sense that it has at most $k$ significant
coordinates, i.e., entries of large magnitude compared to the other entries.
The observation vector $\bfy$ is formed as a linear transformation of $\bfx$,
i.e.,
  $$
   \bfy=\Phi\bfx+\bfz,
  $$
where $\Phi$ is an $m\times N$ real matrix, $m\ll N,$ and $\bfz$ is a noise
vector. We assume that $\bfz$ has bounded energy (i.e., $\|\bfz\|_{2} <\varepsilon$).  
The objective of the estimator is to find a good approximation of
the signal $\bfx$ after observing $\bfy$. This is obviously impossible for general
signals $\bfx$ but becomes tractable if we seek a sparse approximation $\hat\bfx$
which satisfies
   \begin{equation}\label{eq:p1p2}
  \|\bfx-\hat\bfx\|_{p}\le C_1 \min_{\bfx' \text{ is }
  k\text{-sparse} }\|\bfx- \bfx' \|_{q} + C_2 \varepsilon
  \end{equation}
for some $p,q \ge 1$ and constants $C_1,C_2$.
  Note that if $\bfx$ itself is $k$-sparse,
then (\ref{eq:p1p2}) implies that the recovery error $\|\hat\bfx-\bfx\|$ is at most proportional to the norm of the noise.
Moreover it implies that the recovery is stable in the sense that
 if $\bfx$ is approximately $k$-sparse
 then the recovery error is small. If the estimate satisfies an inequality of the type \eqref{eq:p1p2},
 we say that the recovery procedure satisfies a $(p,q)$ error guarantee.

Among the most studied estimators is the Basis Pursuit algorithm
\cite{che98}.
This is an $\ell_1$-minimization algorithm that provides an
estimate of the signal through solving a convex programming problem
  \begin{equation}\label{eq:bp}
   \hat\bfx=\arg\min\|\tilde\bfx\|_{1} \text{\quad subject to }
   \|\Phi\tilde\bfx-\bfy\|_{2}\le \varepsilon.
  \end{equation}
Basis Pursuit is known to provide both $(\ellone,\ellone)$ and $(\elltwo,\ellone)$
error guarantees under the conditions on $\Phi$ discussed in the next section.

Another popular estimator for which the recovery guarantees are
proved using coherence properties of the sampling matrix $\Phi$ is
Lasso \cite{tib96,che98}. Assume the vector $\bfz$ is independent of
the signal and formed of independent
identically distributed Gaussian random variables with zero mean and variance $\sigma^2.$
Lasso is a regularization of the $\ell_0$ minimization problem
written as follows:
  \begin{equation}\label{eq:Lasso}
   \hat\bfx=\arg\min_{\tilde x\in\reals^N}\frac12\|\Phi\tilde \bfx-\bfy\|^2_{2}+
       \lambda_N\sigma^2\|\tilde \bfx\|_{\ell_1}.
  \end{equation}
Here $\lambda_N$ is a regularization parameter which controls the
complexity (sparsity) of the optimizer.

Compressed sensing is just one of a large group of applications of solutions
of severely ill-defined problems under the sparsity assumption. An extensive
recent overview of such applications is given in \cite{bru09}. It is this
multitude of concrete applications that makes the study of sparse recovery
such an appealing area of signal processing and applied statistics.

\subsection{Properties of sampling matrices}
One of the main questions related to sparse recovery is derivation of
sufficient conditions for the convergence and error guarantees of the reconstruction algorithms.
Here we discuss some properties of sampling matrices that are relevant to
our results, focusing on incoherence and near-isometry of random submatrices
of the sampling matrix.

Let $\Phi$ be an $m\times N$ real matrix and let $\phi_1,\dots,\phi_N$ be its
columns.
Without loss of generality throughout this paper we assume that
the columns are unit-length vectors. Let $[N]=\{1,2,\dots, N\}$ and
let $I=\{i_1,\dots,i_k\}\subset [N]$ be a $k$-subset of the set of coordinates. 
By $\cP_k(N)$ we denote the set of all $k$-subsets of $[N].$ Below
we write $\Phi_I$ to refer to the $m\times k$ submatrix of $\Phi$
formed of the columns with indices in $I$. Given a vector
$\bfx\in \reals^N,$ we denote by $\bfx_I$ a $k$-dimensional vector
given by the projection of the vector $\bfx$ on the coordinates
in $I$. 

It is known that at least $m=\Omega(k\log(N/k))$ samples are required
for any recovery algorithm with an error guarantee of the form \eqref{eq:p1p2}
(see for example \cite{kas07,kha10}). Matrices with random Gaussian or Bernoulli
entries with high probability provide the best known error guarantees from the
sketch dimension that matches this lower bound \cite{can05,can06a,can06c}.
The estimates become more conservative once we try to construct
sampling matrices explicitly.

We say that $\Phi$ satisfies the
{\em coherence property} if the inner product $|\ip{\phi_1}{\phi_j}|$
is uniformly small, and call $\mu=\max_{i\ne j}|\ip{\phi_i}{\phi_j}|$ the coherence
parameter of the matrix. 
The importance of incoherent dictionaries has been recognized
in a large number of papers on compressed sensing, among them 
\cite{tro04,tro08,gur09,can07b,can09a,can10a,cai09}. The coherence condition plays an essential
role in proofs of recovery guarantees in these and many other studies. 
We also define the mean square coherence and the maximum average square
coherence of the dictionary:
  $$
  \bar\mu^2=\frac 1{N(N-1)}\sum_{\substack{i,j=1\\i\ne j}}^n\mu^2_{ij}, \quad \bar\mu^2_{\max}=\max_{1\le j\le N}\frac1{N-1}
  \sum_{\substack{i=1\\i\ne j}}^n\mu_{ij}^2.
  $$
Of course, $\bar\mu^2\le \bar\mu^2_{\max}$ with equality if and only if for every $j$ the sum in $\bar\mu^2_{\max}$
takes the same value. 
Our reliance on two coherence parameters of the
sampling matrix $\Phi$ resembles somewhat the approach in 
\cite{baj10a,baj11}; however, unlike those papers, our results imply recovery guarantees for Basis Pursuit.
Our proof methods are also materially different from these works. More details are provided below in 
this section where we comment on previous results.

\subsubsection{The RIP property}
The matrix $\Phi$ satisfies the RIP property (is $(k,\delta)$-RIP)
if
  \begin{equation}\label{eq:RIP}
    (1-\delta)\|\bfx\|_{2}^2\le\|\Phi\bfx\|_{2}^2\le
               (1+\delta)\|\bfx\|_{2}^2
  \end{equation}
holds for all $k$-sparse vectors $\bfx$, where $\delta\in(0,1)$
is a parameter. Equivalently, $\Phi$ is $(k,\delta)$-RIP if $\|\Phi_I^T\Phi_I - \text{Id}\| \le \delta$ holds
for all $I\in [N], |I| = k$, where $\| \cdot \|$ is the spectral norm and $\text{Id}$ is the identity matrix.
The RIP property provides a sufficient condition for the solution of \eqref{eq:bp}
to satisfy the error guarantees of Basis Pursuit \cite{can05,can06b,can08a,cai09}.
In particular, by \cite{can08a}, $(2k,\sqrt 2-1)$-RIP suffices
for both $(\ell_1,\ell_1)$ and $(\ell_2,\ell_1)$ error estimates, while
\cite{cai09} improves this to $(1.75k,\sqrt 2-1)$-RIP.

As is well known (see \cite{tro04}
\cite{don03}), coherence and RIP are related: a matrix with coherence
parameter $\mu$ is $(k, (k-1)\mu)$-RIP. This connection has served the
starting point in a number of studies on constructing RIP matrices from
incoherent dictionaries. To implement this idea one starts with
a set of unit vectors $\phi_1,\dots,\phi_N$ with maximum coherence $\mu.$ In other words, we seek a well-separated collection
of lines through the origin in $\reals^m$, or reformulating again, a
good packing of the real projective space $\reals P^{m-1}.$
One way of constructing such packings begins with taking a set $\cC$ of binary
$m$-dimensional vectors whose pairwise Hamming distances are concentrated around $m/2.$
Call the maximum deviation from $m/2$ the {\em width} $w$ of the set $\cC.$
An incoherent dictionary is obtained by mapping the bits of a small-width code to bipolar
signals and normalizing. The resulting coherence and width are related by $w(\cC)=\mu m/2.$

One of the first papers to put forward the idea of constructing
RIP matrices from binary vectors was the work by DeVore \cite{dev07}.
While \cite{dev07} did not make a connection to error-correcting codes,
a number of later papers pursued both its algorithmic and constructive
aspects \cite{bar10,cal10a,cal10b,dai09}.
Examples of codes with small width are given in \cite{alo92}, where they are
studied under the name of small-bias probability spaces.
RIP matrices obtained from the constructions in \cite{alo92} satisfy
$m=O(\frac{k\log N}{\log(\log k N)})^2$.
Ben-Aroya and Ta-Shma \cite{ben09}
recently improved this to $m=O(\frac{k\log N}
{\log k})^{5/4}$ for $(\log N)^{-3/2}\le \mu\le (\log N)^{-1/2}.$ 
The advantage of obtaining RIP matrices from binary or spherical codes
is low construction complexity: in many instances it is possible to define
the matrix using only $O(\log N)$ columns while the remaining columns can
be computed as their linear combinations.
We also note a result by Bourgain et al. \cite{bou10} who
gave the first (and the only known) construction of RIP matrices with $k$ on the
order of $m^{\frac 12+\epsilon}$
(i.e., greater than $O(\sqrt m)$). An overview of the
state of the art in the construction of RIP matrices is given in
a recent paper \cite{bandeira12}. 

 At the same time, in practical
problems we still need to write out the entire matrix; so constructions
of complexity $O(N)$ are an acceptable choice. Under these assumptions,
the best tradeoff between $m,k$ and $N$ for RIP-matrices based on codes and coherence is obtained
from Gilbert-Varshamov type code constructions: namely, it is possible
to construct $(k,\delta)$-RIP matrices with $m=4(k/\delta)^2\log N$.
At the same time, already \cite{alo92} observes that
the sketch dimension in RIP matrices constructed from binary codes is
at least $m=\Theta((k^2\log N)/\log k).$
\subsubsection{Statistical incoherence properties}
The limitations on incoherent dictionaries discussed in the previous
section suggest relaxing the RIP condition. An intuitively appealing idea
is to require that condition \eqref{eq:RIP} hold for almost all
rather than all $k$-subsets $I,$ replacing RIP with a version
of it, in which the near-isometry property holds with high probability with respect 
to the choice of $I\in\cP_k(N).$ Statistical RIP (StRIP) matrices are arguably easier to construct, so
they have a potential of supporting provable recovery guarantees
from shorter sketches compared to the known constructive schemes
relying on RIP.

A few words on notation. Let $[ N]:=\{1,2,\dots,N\}$ and let $\cP_k(N)$ denote the set of $k$-subsets of $[N].$ The usual notation 
for probability $\Pr$ is used to refer a probability measure when there is no ambiguity. At the same time, we use
separate notation for some frequently encountered probability spaces.
In particular, we use $P_{R_k}$ to denote the uniform probability distribution on $\cP_k(N)$.
If we need to choose a random $k$-subset $I$ and a random index in $[N]\backslash I,$ we use the notation $P_{R_k'}$. 
We use $P_{R^k}$ to denote any probability  measure on $\mathbb{R}^k$ which assigns equal probability to each of the $2^k$ 
orthants (i.e., with uniformly distributed signs).

The following definition is essentially due to Tropp \cite{tro08,tro08b}, where it is called
conditioning of random subdictionaries.
\begin{definition}\label{def:strip} An $m\times N$  matrix $\Phi$ satisfies the statistical 
RIP property (is $(k,\delta,\epsilon)$-StRIP)
if 
   $$
    P_{R_k}(\{I\in\cP_k(N):\|\Phi^T_I\Phi_I-\text{\rm Id}\|\le \delta\})\ge 1-\epsilon.
   $$
   In other words,
the inequality
  \begin{equation}\label{eq:strip}
   (1-\delta)\|\bfx\|_{2}^2\le \|\Phi_I \bfx\|^2\le (1+\delta)\|\bfx\|_{2}^2
   \end{equation}
holds for at least a $1-\epsilon$ proportion of all $k$-subsets of $[N]$
and for all $x\in \reals^k.$
\end{definition}
A related but different definition was given later in several papers such as
\cite{cal10a,baj10a,gur09} as well as some others. In these works, a matrix is called
$(k,\delta,\epsilon)$-StRIP if inequality \eqref{eq:strip} holds for at least $1-\epsilon$ proportion of $k$-sparse unit vectors $\bfz\in \reals^N$. While several well-known classes of matrices were shown to have this property, it is not
sufficient for sparse recovery procedures. Several additional properties as well as specialized recovery procedures
that make signal reconstruction possible were investigated in \cite{cal10a}.

\remove{\begin{definition}\label{def:strip1} An $m\times N$ matrix $\Phi$ is said to be $(k,\delta,\epsilon)$ StRIP
if inequality \eqref{eq:strip} holds for at least $1-\epsilon$ proportion of $k$-sparse unit vectors $\bfz\in \reals^N$. 
By choosing first the support and then the coordinates of $\bfx$ we can write
this condition as
  $$
    P_{R_k\times{S^{k-1}}}((I,\bfx): \,(1-\delta)\|x\|_{2}^2\le \|\Phi_I x\|^2\le (1+\delta)\|x\|_{2}^2)\ge 1-\epsilon.
  $$
The probability $P_{S^{k-1}}$ refers to the uniform distribution on the set of unit-norm vectors in $\reals^k.$
\end{definition}}

In this paper we focus on the statistical isometry property as given by Def. \ref{def:strip} and mean
this definition whenever we mention StRIP matrices. We note that condition \eqref{eq:strip} is scalable, 
so the restriction to unit vectors is not essential.
%
\begin{definition} \label{def:sinc}
An $m\times N$ matrix $\Phi$ satisfies a statistical incoherence condition (is $(k,\alpha,\epsilon)$-SINC)
if 
  \begin{equation}\label{eq:sinc}
    P_{R_k}(\{I\in \textstyle{\cP_k(N)}: \max_{i\not\in I}\|\Phi_I^T\phi_i\|_2^2\le\alpha\})\ge 1-\epsilon.
  \end{equation}
\end{definition}
This condition is discussed in \cite{fuc04,tro05}, and more explicitly in \cite{tro08b}.
Following \cite{tro08b}, it appears in the proofs of sparse
recovery in \cite{can09a} and below in this paper. 
A somewhat similar average coherence condition was also
introduced in \cite{baj10a,baj11}. The reason that \eqref{eq:sinc}
is less restrictive than the coherence property is as follows.
Collections of unit vectors with small coherence (large separation)
cannot be too large so as not to contradict universal bounds on packings of
$\reals P^{m-1}.$ At the same time, for the norm
$\|\Phi_I^T\phi_i\|_2$ to be large it is necessary that
a given column is close to the majority of the $k$ vectors from
the set $I$, which is easier to rule out.

Nevertheless, the above relaxed conditions are still restrictive enough to rule out many
deterministic matrices: the problem is that for almost all supports $I$ 
we require that $\|\Phi_I\phi_i\|$ be small for all $i\not\in I.$
We observe that this condition can be further relaxed. Namely, let 
   $$
   \cB(\Phi)=\{t\in \reals:\, \exists I\in{\textstyle{\cP_k(N)}}, \, i\in I^c \text{ such that }\|\Phi_I^T\phi_i\|_2=t\}
   $$
be the set of all values taken by the coherence parameter.
Let us introduce the following definition.
\begin{definition}\label{def:wsinc}
An $m\times N$ matrix $\Phi$ is said to satisfy a weak statistical incoherence condition
(to be a $(k,\delta,\alpha,\epsilon)$-WSINC) if 
    \begin{equation}\label{eq:wsinc}
    \sum_{t\in\cB(\Phi)}  P_{R_k'}(\{(I,i), I\in A_\alpha(\Phi), i\in I^c \text{ such that }\|\Phi_I^T\phi_i\|_2=t\})g(\delta,t)\le 
  \frac{\epsilon}{N-k},
    \end{equation}
where $g(\delta,t)$ is a positive increasing function of $t$ and
   $$
   A_\alpha(\Phi)=\{I\in \textstyle{\cP_k(N)}: \exists i\in I^c\text{ such that }\|\Phi_I^T\phi_i\|_2^2>\alpha\}.
   $$
 \end{definition}   
We note that this definition is informative if $g(\delta,t)\le 1;$ otherwise, replacing it with 1 we get back 
the SINC condition. Below we use $g(\delta,t)=\exp(-(1-\delta)^2/(8t^2)).$
This definition takes account of the distribution of values of the quantity $\|\Phi_I^T\phi_i\|$
for different choices of the support and a column $\phi_i$ outside it. For binary dictionaries,
the WSINC property relies on a distribution of sums of Hamming distances between a column and a collection
of $k$ columns, taken with weights that decrease as the sum increases.

\begin{definition}
We say that a signal $\bfx\in \reals^N$ is drawn from a {\sl generic random
signal model} $\cS_k$ if

1) The locations of the $k$ coordinates of $\bfx$ with largest magnitudes
are chosen among all $k$-subsets $I\subset [N]$ with a uniform distribution;

2) Conditional on $I$, the signs of the coordinates $x_i, i\in I$ are
i.i.d. uniform Bernoulli random variables taking values in the set $\{1,-1\}$.
\end{definition}
Using previous defined notation, the probability induced by the generic model $P_{\cS_k}$ can be decomposed as $P_{R_k\times  R^k}$.

\vspace*{.1in}
\subsection{Contributions of this paper} Our results are as follows. First, we show that a combination of the StRIP and SINC
conditions suffices for stable recovery of sparse signals. In their large part, these results 
are due to \cite{tro08}. We incorporate some additional elements such as stability analysis of Basis Pursuit
based on these assumptions and give the explicit values of the constants involved in the
assumptions. We also show that the WSINC condition together with StRIP is sufficient for bounding
the off-support error of Basis Pursuit.

One of the main results of \cite{tro08,tro08b} is a 
sufficient condition for a matrix to act nearly isometrically on most sparse vectors. Namely,
an $m\times N$ matrix $\Phi$ is $(k,\delta,\epsilon=k^{-s})$-StRIP if
  $$
 \sqrt{  s\mu^2 k\log(k+1)}+\frac{k}{N}\|\Phi\|^2\le c\delta,
  $$
where $s\ge 1$ and $c$ is a constant; see \cite{tro08}, Theorem B. 
For this condition to be applicable, one needs that $\mu=O(1/\sqrt{k\log (1/\epsilon)}).$  
For sampling matrices that satisfy this condition, we obtain a near-optimal relation $m=O(k\log(N/\epsilon))$
between the parameters. Some examples of this kind are given below in Sect.~\ref{sec:determin}.
As one of our main results, we extend the region of parameters that suffice 
for $(k,\delta,\epsilon)$-StRIP. Namely, in Theorem \ref{thm:1/4} we prove that it is enough to have 
the relation $\mu=O(1/\sqrt[4]{k\log k\log^3(1/\epsilon)})$.
This improvement comes at the expense of an additional requirement on $\bar\mu^2=O(1/(k\log(1/\epsilon)))$ (or 
a similar inequality for $\bar\mu_{\max}^2$), but this is easily satisfied in a large class of examples, discussed
below in the paper. 
These examples in conjunction with Theorem \ref{thm:sinc} and the results in 
Section \ref{sect:lp-decoding} establish provable error guarantees for some new classes of sampling matrices.

We note a group of papers by Bajwa and Calderbank \cite{baj10a,baj11,cal10b} which is centered around
the analysis of a threshold decoding procedure (OST) defined in \cite{baj10a}. The sufficient conditions in 
these works are formulated in terms of $\mu$ and maximum average coherence 
$\nu=\frac1{N-1}\max_{1\le j\le N}|\sum_{i\ne j}\langle \phi_i,\phi_j\rangle|.$ 
Reliance on two coherence parameters of $\Phi$ for establishing sufficient conditions for error estimates in \cite{baj10a} 
is a shared feature of these papers and our research.
At the same time, the OST procedure relies on additional assumptions such as minimum-to-average ratio of signal
components bounded away from zero (in experiments, OST is efficient for uniform-looking signals,
and is less so for sparse signals with occasional small components). Some other similar assumptions
are required for the proofs of the noisy version of OST \cite{baj11}. 

We note that there is a number of other studies that establish sufficient conditions for
sampling matrices to provide bounded-error approximations in sparse recovery procedures, e.g., \cite{can10a,jud11a,jud11b}.
At the same time, these conditions are formulated in terms different from our
assumptions, so no immediate comparison can be made with our results. 

As a side result, we also calculate the parameters for the StRIP and SINC conditions that suffice to derive an error estimate for
sparse recovery using Lasso. This result is implicit in the work of Cand{\'e}s and Plan \cite{can09a}, which also 
uses the SINC property of sampling matrices.
The condition on sparsity for Lasso is in the form
$k=O(N/\|\Phi\|^2\log N),$ so if $\|\Phi\|^2\approx N/m,$ this yields $k\le O(m/\log N).$
This range of parameters exceeds the range in which Basis Pursuit is shown to have good error guarantees,
even with the improvement obtained in our paper. At the same time, both \cite{can09a} and our
calculations find error estimates in the form of bounds on $\|\Phi \bfx-\Phi \hat \bfx\|_2$ rather than
$\|\bfx-\hat\bfx\|_2,$ i.e., on the compressed version of the recovered signal.

In the final section of the paper we collect examples of incoherent dictionaries that satisfy our 
sufficient conditions for approximate recovery using Basis Pursuit. Two new examples with nearly optimal parameters
that emerge are the Delsarte-Goethals dictionaries \cite{mac91} and deterministic sub-Fourier dictionaries \cite{hau10}.
For instance, in the Delsarte-Goethals case we obtain the sketch dimension $m$ on the order of $k\log^3\frac{N}{\epsilon},$ 
which is near-optimal, and is in line with the comments made above.

We also show that the restricted independence property of the dictionary suffices to establish the StRIP
condition. Using sets of binary vectors known as orthogonal arrays, we find $(k,\delta,\epsilon)$-StRIP dictionaries with 
$k=O(m^{3/7}).$ At the same time, we are not able to show that restricted independence gives rise to the
SINC property with good parameter estimates, so this result has no consequences for linear programming decoders.

\vspace*{.05in}{\em Acknowledgment:} We are grateful to Waheed Bajwa for useful feedback on an early version of this work.

\section{Statistical Incoherence Properties and Basis Pursuit}\label{sect:lp-decoding}
In this section we prove approximation error bounds for recovery by
Basis Pursuit from linear sketches obtained using
deterministic matrices with the StRIP and SINC properties.

\subsection{StRIP Matrices with incoherence property} It was proved in \cite{tro08} that random 
sparse signals sampled using matrices with the StRIP property can be 
recovered with high probability from low-dimensional sketches using linear programming.
In this section we prove a similar result that in addition incorporates
stability analysis. 
\begin{theorem}\label{theorem:bp}
Suppose that $\bfx$ is a generic random signal from the model
$\cS_k.$ Let
$\bfy=\Phi\bfx$ and let $\hat\bfx$ be the approximation of $\bfx$ by the Basis
Pursuit algorithm.
Let $I$ be the set of $k$ largest coordinates of $\bfx$. If
  \begin{enumerate}
  \item $\Phi$ is $(k,\delta,\epsilon)$-StRIP;
  \item $\Phi$ is $(k,\frac{(1-\delta)^2}{8\log(2N/\epsilon)},\epsilon)$-SINC,
\end{enumerate}
then with probability at least $1-3\epsilon$
  $$
  \|\bfx_I-\hat\bfx_I\|_{2}\le \frac{1}{2\sqrt{2\log(2N/\epsilon)}}
      \min_{\bfx'\text{\rm is $k$
  -sparse} }\|\bfx- \bfx' \|_{1}
  $$
and
  $$
\|\bfx_{I^c}-\hat{\bfx}_{I^c}\|_{1} \le 4 \min_{\bfx' \text{\rm is $k$
  -sparse} }\|\bfx- \bfx' \|_{1}
  $$
\end{theorem}
This theorem implies that if the signal $\bfx$ itself is $k$-sparse then
the basis pursuit algorithm will recover it exactly. Otherwise, its output
$\hat \bfx$ will be a tight sparse approximation of $\bfx$.

Theorem \ref{theorem:bp} will follow from the next three lemmas. Some of the
ideas involved in their proofs are close to the techniques used in \cite{can06a}.
Let $\bfh=\bfx-\hat\bfx$ be the error in recovery of basis pursuit.
In the following $I\subset[N]$ refers to the support of the $k$ largest coordinates
of $\bfx.$ 
\begin{lemma}\label{lemma:error-sup} Let $s=8\log(2N/\epsilon).$
Suppose that
   $\|(\Phi_I^T\Phi_I)^{-1}\|\le\frac1{1-\delta}$
and
  $$
   \|\Phi_I^T\phi_i\|_{2}^2\le s^{-1}(1-\delta)^2 \quad\text{for all }i\in I^c := [N]\setminus I.
  $$
Then
  $$
  \|\bfh_I\|_{2}\le s^{-\half}\,\|\bfh_{I^c}\|_{1}.
  $$
\end{lemma}
\begin{proof} Clearly,
  $
   \Phi\bfh=\Phi\hat\bfx-\Phi\bfx=0,
  $
so $\Phi_I\bfh_I=-\Phi_{I^c}\bfh_{I^c}$ and
  $$
  \bfh_I=-(\Phi_I^T\Phi_I)^{-1}\Phi_I^T\Phi_{I^c}\bfh_{I^c}.
  $$
We obtain
  \begin{align*}
    \|\bfh_I\|_{2}&\le \|(\Phi_I^T\Phi_I)^{-1}\|\|\Phi_I^T\Phi_{I^c}\bfh_{I^c}\|_{2}\le
   \frac1{1-\delta}\sum_{i\in I^c}\|\Phi_I^T\phi_i\|_{2}|h_i|\\
  &\le  s^{-\half}\,\|\bfh_{I^c}\|_{1},
 \end{align*}
as required.
\end{proof}

Next we show that the error outside $I$ cannot be large. Below $\sgn(\bfu)$ is a $\pm1$-vector
of signs of the argument vector $\bfu.$
\begin{lemma} \label{lemma:v}
  Suppose that there exists a vector $\bfv\in\reals^N$ such that
  \begin{enumerate}
   \item[(i)] $\bfv$ is contained in the row space of $\Phi$, say $\bfv=\Phi^T\bfw;$
   \item[(ii)] $\bfv_I=\sgn(\bfx_I);$
   \item[(iii)] $\|\bfv_{I^c}\|_\ellinf\le \half.$
  \end{enumerate}
Then
  \begin{equation}\label{eq:h}
    \|\bfh_{I^c}\|_{1}\le 4\|\bfx_{I^c}\|_{1}.
  \end{equation}
\end{lemma}
\begin{proof}
By \eqref{eq:bp} we have
  \begin{align*}
    \|\bfx\|_{1}&\ge \|\hat\bfx\|_{1}=\|\bfx+\bfh\|_{1}
   =\|\bfx_I+\bfh_I\|_{1}+\|\bfx_{I^c}+\bfh_{I^c}\|_{1}\\
    &\ge \|\bfx_I\|_{1}+\ip{\sgn(\bfx_I)}{\bfh_I}+\|\bfh_{I^c}\|_{1}-\|\bfx_{I^c}\|_{1}.
  \end{align*}
Here we have used the inequality $\|\bfa+\bfb\|_{1}\ge\|\bfa\|_{1}+\ip{\sgn(\bfa)}
{\bfb}$ valid for any two vectors $\bfa,\bfb\in\reals^N$ and the triangle inequality.
From this we obtain
  $$
  \|\bfh_{I^c}\|_{1}\le |\ip{\sgn(\bfx_I)}{\bfh_I}|+2\|\bfx_{I^c}\|_{1}.
  $$
Further, using the properties of $\bfv,$ we have
\begin{eqnarray*}
|\langle\sgn(\bfx_I),\bfh_I\rangle| &=& |\langle\bfv_I,\bfh_I\rangle|\\
&=& |\langle\bfv,\bfh\rangle - \langle\bfv_{I^c},\bfh_{I^c}\rangle|\\
&\le& |\langle \Phi^T \bfw,\bfh\rangle| + |\langle\bfv_{I^c},\bfh_{I^c}\rangle|\\
&\le& |\langle \bfw,\Phi \bfh\rangle| + \|\bfv_{I^c}\|_{\ell_\infty} \|\bfh_{I^c}\|_{1}\\
&\le& \frac12  \|\bfh_{I^c}\|_{1}.
\end{eqnarray*}
The statement of the lemma is now evident.
\end{proof}

Now we prove that such a vector $\bfv$ as defined in the last lemma indeed exists.
\begin{lemma} \label{lemma:exists_v} Let $\bfx$ be a generic random signal from the model $\cS_k.$
Suppose that the support $I$ of the $k$ largest coordinates of $\bfx$ is fixed.
Under the assumptions of Lemma \ref{lemma:error-sup} the vector
  $$
   \bfv=\Phi^T\Phi_I(\Phi_I^T\Phi_I)^{-1}\sgn(\bfx_I)
  $$
satisfies (i)-(iii) of Lemma \ref{lemma:v} with probability at least $1-\epsilon.$ 

\end{lemma}
\begin{proof}
From the definition of $\bfv$ it is clear that it belongs to the row-space of $\Phi$ and
$\bfv_I =\sgn(\bfx_I).$
We have $v_i = \phi_i^T \Phi_I (\Phi_I^T\Phi_I)^{-1}\sgn(\bfx_I) = \ip{\bfs_i}{\sgn(\bfx_I)},$
where
   $$
\bfs_i = (\Phi_I^T\Phi_I)^{-1}\Phi_I^T \phi_i \in \reals^k.
   $$
We will show that $|v_i|\le \frac12$ for all $i \in I^c$ with
probability $1-\epsilon.$

Since the coordinates of $\sgn(\bfx_I)$ are i.i.d. uniform random variables taking
values in the set $\{\pm1\}$, we can use Hoeffding's inequality to claim that
\begin{equation}\label{Hoeffding}
P_{R^k}(|v_i| >1/2 )\le 2\exp\Big(-\frac1{8\|\bfs\|_2^2}\Big).
\end{equation}
On the other hand, for all $i\in I^c,$
\begin{eqnarray}\label{magnitude}
\|\bfs_i\|_{2} &=& \|(\Phi_I^T\Phi_I)^{-1}\Phi_I^T \phi_i\|_{2}\notag\\
&\le & \|(\Phi_I^T\Phi_I)^{-1}\| \|\Phi_I^T \phi_i\|_{2}\notag\\
&\le& \frac1{1-\delta} \frac{1-\delta}{\sqrt{8\log(2N/\epsilon)}}\notag\\
&=& \frac1{\sqrt{8\log(2N/\epsilon)}}.
\end{eqnarray}
Equations \eqref{Hoeffding} and \eqref{magnitude} together imply for any $i\in I^c,$
$$
P_{R^k}\Big(|v_i| >\frac12\Big) \le 2\exp\Big(-\frac1{8(1/\sqrt{8\log(2N/\epsilon)})^2}\Big)= \frac{\epsilon}N.
$$
Using the union bound, we now obtain the following relation:
   \begin{equation}\label{eq:1/2}
 P_{R^k}\Big(\|\bfv_{I^c}\|_{\infty}>1/2\Big)\le \epsilon.
   \end{equation}
Hence $|v_i|\le \frac12$ for all $i \in I^c$ with probability at least $1-\epsilon$.
\end{proof}

Now we are ready to prove Theorem~\ref{theorem:bp}.
\vspace{.1in}

\begin{proof}[Proof of Theorem~\ref{theorem:bp}]
The matrix $\Phi$ is $(k,\delta,\epsilon)$-SRIP. Hence,
with probability at least $1-\epsilon,$ $\|(\Phi_I^T\Phi_I)^{-1}\| \le \frac{1}{1-\delta}$.
At the same time, from the SINC assumption we have, with probability at least $1-\epsilon$
over the choice of $I$,
  $$
\|\Phi_I^T\phi_i\|_{2}^2 \le \frac{(1-\delta)^2}{8\log(2N/\epsilon)},
  $$
for all $i \in I^c.$ Thus, $\Phi_I$ will have these two properties
with probability at least $1-2\epsilon$.
Then from Lemma~\ref{lemma:error-sup} we obtain that
 $$
 \|\bfh_I\|_2 \le \frac{1}{\sqrt{8\log(2N/\epsilon)}} \|\bfh_{I^c}\|_1,
 $$
with probability $\ge 1-2\epsilon.$ Furthermore, from Lemmas \ref{lemma:v}, \ref{lemma:exists_v}
$$
\|\bfh_{I^c}\|_1 \le 4 \|\bfx_{I^c}\|_1,
$$
with probability $1-\epsilon$. This completes the proof.
  \end{proof}

\subsection{StRIP Matrices with weak incoherence property}
In this section we establish a recovery guarantee of Basis Pursuit under the weak SINC condition defined earlier in the paper.
\begin{theorem}\label{thm:wsinc}
 Suppose that the sampling matrix $\Phi$ is $(k,\delta,\epsilon)$-StRIP and $\big(k,\delta,\alpha,\epsilon^2\big)$-WSINC,
 where $\alpha={(1-\delta)^2}/{8\log (2N/\epsilon)}$ and $g_\delta(t)=\exp(-(1-\delta)^2/8t^2).$ 
Suppose that the signal $\bfx$ is chosen from the generic random signal model and let $\hat\bfx$ be 
the approximation of $\bfx$ found by Basis Pursuit. Then 
with probability at least $1-4\epsilon$ we have
  $$
  \|\bfx_{I^c}-\hat\bfx_{I^c}\|_{{1}}\le 4\min\limits_{\bfx' \text{\rm is $k$-sparse}} \|\bfx-\bfx'\|_{1}.
  $$
\end{theorem}
If $\bfx$ is $k$-sparse and satisfies the condition $\bfy=\Phi \bfx$, then this theorem asserts that Basis Pursuit will 
find the support of 
$x$. If in addition $\bfx$ is the only $k$-sparse solution to $\bfy=\Phi \bfx,$ 
then we have $\hat{x}=x$. Note that the WSINC property is not sufficient for
the $(\elltwo,\ellone)$ error guarantee.
 However, once the corrected support is detected, the signal $\bfx$ can be 
found by solving the overcomplete system $\bfy=\Phi_I \bfx$.

To prove Theorem \ref{thm:wsinc}, we refine the ideas used to establish Lemma
\ref{lemma:exists_v}.
\begin{lemma}\label{lemma:wsinc}
Suppose that the sampling matrix $\Phi$ satisfies the conditions of Theorem \ref{thm:wsinc}. For any $\bfx\in\reals^k$
and $I\subset[N]$ define $v(\bfx,I)=\Phi^T\Phi_I(\Phi_I^T\Phi_I)^{-1}\sgn(\bfx).$ 
Let 
   $$
       p(I)=P_{R^k}(\|v_{I^c}(\bfx,I)\|_\infty>1/2),
   $$
Then
   $$
    P_{R_k}(\{I: p(I)>\epsilon\})<3\epsilon.
   $$
\end{lemma}
\begin{proof}
As in the proof of Lemma \ref{lemma:exists_v}, we define the vector
   $$
\bfs_i(I) = (\Phi_I^T\Phi_I)^{-1}\Phi_I^T \phi_i \in \reals^k
   $$
   and let $v_i(\bfx,I)$ be the $i$th coordinate of the vector $v(\bfx,I).$ 
From now on we write simply $v_i,\bfs_i,$ omitting the dependence on $I$ and $\bfx$.
Let $M=M(\Phi):=\{I\in \cP_k(N):   \|\Phi_I^T\Phi_I\|_2\ge 1-\delta\},$ then the StRIP property of $\Phi$ implies that
  $$
  P_{R_k}(M)\ge 1-\epsilon.
  $$
By definition, for any $I\in M$
  $$
  \|\bfs_i\|_2 = \|(\Phi_I^T\Phi_I)^{-1}\Phi^T_I\phi_i\|_2 \leq \frac{1}{1-\delta}\|\Phi^T_I\phi_i\|_2.
  $$
Now we split the target probability into three parts:
  \begin{align*}
  P_{R_k}(\{I: p(I)>\epsilon\})&=P_{R_k}(\{I\in M\cap A: p(I)>\epsilon\})+
    P_{R_k}(\{I\in M\cap A^c: p(I)>\epsilon\})\\
    &+ P_{R_k}(\{I\in M^c: p(I)> \epsilon\}),
  \end{align*}
where $A=A_\alpha(\Phi)=\{ I : \|\Phi^T_{ I }\phi_i\|^2_2 > \alpha \text{ for some } i\in  I^c \}$ is the 
set of supports appearing in the definition of the WSINC property. If $I\in M\cap A,$ i.e., it supports both
StRIP and SINC properties, then \eqref{eq:1/2} implies that $p(I)\le\epsilon,$ so the first term on the right-hand side equals 0. 
The third term refers to supports with no SINC property, whose total probability is $\le \epsilon.$
 Estimating the second term
by the Markov inequality, we have
  \begin{equation}\label{eq:M}
  P_{R_k}(\{ I \in M \cap A^c: p(I)>\epsilon\} )
\leq \frac{\avg_{R_k}[p(I), {\bf 1}(I \in M\cap A^c)]}{\epsilon}
   \end{equation}
where  ${\bf 1}(\cdot)$ denotes the indicator random variable. We have
  \begin{equation}\label{eq:E}
  \avg_{R_k}[p(I), I \in M\cap A^c]=\avg_{R_k}[p(I){\bf 1}( I \in M\cap A^c)]= 
  \sum_{I\in M\cap A^c}\frac1{\binom Nk}p(I),
  \end{equation}

Let us first estimate $p(I)$ for $I\in M\cap A^c$ by invoking Hoeffding's inequality \eqref{Hoeffding}:
\begin{align*}
  p(I)&= P_{R^k}(\exists i\in I^c, \ \ |v_i|>1/2)\leq \sum_{i\in I^c}   P_{R^k}(|v_i|>1/2 ) \\
  &\leq \sum_{i\in I^c} 2\exp \Big (-\frac 1{8\|\bfs_i\|_2^2}\Big) \\
&\stackrel{\eqref{magnitude}}\leq\sum_{i\in I^c} 2\exp \Big(-\frac {(1-\delta)^2}{8\|\Phi^T_I\phi_i\|_2^2}\Big)\\
& = 2(N-k)\sum_{t\in\cB(\Phi)}  \exp\Big(-\frac{(1-\delta)^2}{8t^2}\Big)
   P_{R'_k}(\|\Phi_{I}^T\phi_i\|_2=t\mid I).
\end{align*}
Substituting this result into \eqref{eq:E}, we obtain 
  \begin{align*}
  \avg_{R_k}[p(I), \{I \in M\cap A^c\}]&\le 2(N-k)\sum_{t\in\cB(\Phi)} \exp\Big(-\frac{(1-\delta)^2}{8t^2}\Big)
    \sum_{I\in M\cap A^c}\frac1{\binom Nk}P_{R'_k}(\|\Phi_{I}^T\phi_i\|=t\mid I)\\
    &{\le} 2(N-k)\sum_{t\in\cB(\Phi)} \exp\Big(-\frac{(1-\delta)^2}{8t^2}\Big) P_{R'_k}(I\in A^c, \|\Phi_I^T\phi\|_2=t)\\
    &\le 2\epsilon^2
  \end{align*}
  where the last step is on account of \eqref{eq:M} and the WSINC assumption.
\end{proof}

{\em Proof of Theorem \ref{thm:wsinc}:} Define the set $B$ by
   $$
   B=\{I\in R_k: P_{{R^{k}}}(\|\bfv_{I^c}\|_{\infty}>1/2\mid I)>\epsilon\}.
   $$
Recall that Theorem \ref{thm:wsinc} is stated with respect to the random signal $\bfx$. Therefore,
let us estimate the probability
 \begin{align*}
 P_{R_k\times {R^{k}}}&( \{(I,\bfx): \,\|\bfv_{I^c}\|_{\infty}>1/2\})\\
 &=\sum_{I\in \cP_k(N)}
 P_{R_k\times{R^{k}}}(\{\bfx:\|\bfv_{I^c}\|_{\infty}>1/2\}\mid I)P_{R_k\times {R^{k}}}(I)\\
 &=\sum_{I\in B^c}P_{{R^{k}}}(\{\bfx:\|\bfv_{I^c}\|_{\infty}>1/2\}\mid I)P_{R_k}(I) +\sum_{I\in B}P_{{R^{k}}}(\{\bfx:\|\bfv_{I^c}\|_{\infty}>1/2\}\mid I)P_{R_k}(I).
 \end{align*}
We have $P_{{R^{k}}}(\{\bfx:\|\bfv_{I^c}\|_{\infty}>1/2\}\mid I)<\epsilon$ from Lemma \ref{lemma:exists_v} and
  $
    P_{R_k}(B)\le 3\epsilon
  $
from Lemma \ref{lemma:wsinc}, so
  $$
  P_{R_k\times {R^{k}}} ( \{(I,\bfx): \,\|\bfv_{I^c}\|_{\infty}>1/2\})<\epsilon(1+3\epsilon)<4\epsilon.
  $$ 
This implies that with probability $1-4\epsilon$ the signal $\bfx$ chosen from the generic random signal
model satisfies the conditions of Lemma \ref{lemma:v}, i.e., 
  $$
  \|\bfx_{I^c}-\hat\bfx_{I^c}\|_1\le 4\|\bfx_{I^c}\|_1.
  $$
This completes the proof. 
\qed

\section{Incoherence Properties and Lasso}
In this section we prove that sparse signals can be approximately recovered
from low-dimensional observations using Lasso if the sampling matrices have
statistical incoherence properties.
The result is a modification of the methods developed in \cite{can09a,tro08}
in that we prove that the conditions used there to bound the error of the Lasso
estimate hold with high probability if $\Phi$ is has both StRIP and SINC properties.
The precise claim is given in the following statement.
\begin{theorem}\label{thm:main2}
Let $\bfx$ be a random $k$-sparse signal  whose support satisfies the two properties of
 the generic random signal model $S_k.$
Denote by $\hat\bfx$  its estimate from $\bfy = \Phi \bfx +\bfz$ via Lasso \eqref{eq:Lasso}, where
$\bfz$ is a i.i.d. Gaussian vector with zero mean and variance $\sigma^2$ and where
$\lambda = 2\sqrt{2\log N}.$
Suppose that $k \le \frac{c_0 N}{\|\Phi\|^2\log N}$,
where $c_0$ is a positive constant, and that
 the matrix $\Phi$ satisfies the following two properties:
\begin{enumerate}
  \item $\Phi$ is $(k,\frac12,\epsilon)$-StRIP.
  \item $\Phi$ is $(k,\frac{1}{128\log(N/2\epsilon)},\epsilon)$-SINC.
\end{enumerate}
Then we have
$$
\|\Phi\bfx -\Phi\hat{\bfx}\|_{2}^2 \le C_0 k\log{N}\sigma^2,
$$
with probability at least $1-3\epsilon - \frac1{N\sqrt{2\pi\log N}}- N^{-a}$,
where $C_0>0$ is an absolute constant and $a=0.15\log(2N/\epsilon)-1.$
\end{theorem}

The following theorem is implicit in \cite{can09a}, see
Theorem 1.2 and Sect 3.2 in that paper.
\begin{theorem}\label{thm:CP}{\rm (Cand{\`e}s and Plan)}
Suppose that $\bfx$ is a $k$-sparse signal drawn from the model $S_k,$
where
  $$
   k \le \frac{c_0
  N}{\|\Phi\|^2\log N},$$  where $c_0>0$ is a constant.
Let $I\subset [N]$ be the support of $\bfx$ and
suppose the following three conditions are satisfied:
\begin{enumerate}
  \item $\|(\Phi_I^T\Phi_I)^{-1}\| \le 2.$
  \item $\|\Phi^T\bfz\|_\ellinf \le 2\sqrt{\log N}.$
  \item $
  \|\Phi_{I^c}^T\Phi_I (\Phi_I^T\Phi_I)^{-1}\Phi_I^T\bfz\|_\ellinf
  + \sqrt{8\log N}\|\Phi_{I^c}^T\Phi_I (\Phi_I^T\Phi_I)^{-1}\sgn(\bfx_I)\|_\ellinf
  \le (2-\sqrt{2})\sqrt{2\log N}.
          $
\end{enumerate}
Then
$$
\|\Phi\bfx -\Phi\hat{\bfx}\|_{2}^2 \le C_0 k(\log{N})\sigma^2,
$$
where $C_0$ is an absolute constant.
\end{theorem}
Our aim will be to prove that conditions (1)-(3) of this theorem hold
with large probability under the assumptions of Theorem \ref{thm:main2}.

First, it is clear that $\|\Phi^T\bfz\|_\infty\le 2\sqrt{\log N}$ with probability
at least $1-(N\sqrt{2\pi\log N})^{-1}.$ This follows simply because $\bfz$ is
an independent Gaussian vector, and has been discussed in \cite{can09a}
(this is also the reason for selecting the particular value of $\lambda_N$).
The main part of the argument is contained in the following lemma whose proof
uses some ideas of \cite{can09a}.
\begin{lemma}\label{lemma:128}
  Suppose that $\half\le\|\Phi^T_I\Phi_I-\text{\rm Id}\|\le\nicefrac32$
and that for all $i \in I^c,$
   $$
    \|\Phi_I^T\phi_i\|_{2}^2\le (128 \log(2N/\epsilon))^{-1}.
  $$
Then Condition (3) of Theorem \ref{thm:CP} holds
with probability at least $1-\epsilon-N^{-a}$ for $a=0.15\log(2N/\epsilon)-1$.
\end{lemma}
\begin{proof}
Let $i\in I^c.$ Define
   $ Z_{0,i} =
\ip{\bfw_i}{ \sgn(\bfx_I)}
$
and $Z_{1,i}=\ip{\bfw_i'}{\bfz},$ where
  \begin{align*}
\bfw_i &=  (\Phi_I^T\Phi_I)^{-1}\Phi_I^T \phi_i,\\
\bfw'_i &= \Phi_I (\Phi_I^T\Phi_I)^{-1}\Phi_I^T \phi_i.
 \end{align*}
Let $Z_0 = \max_{i \in I^c}|Z_{0,i}|$ and
$Z_1 = \max_{i \in  I^c}|Z_{1,i}|.$ We will show that with high probability $Z_0\le 1/4$
and $Z_1\le(1.5-\sqrt 2)\sqrt{2\log N}$ which will imply the lemma.
We compute
     \begin{align*}
\|\bfw_i\|_{2} &\le  \|(\Phi_I^T\Phi_I)^{-1}\|\|\Phi_I^T \phi_i\|_{2}
\le 2 \frac{1}{8\sqrt{2\log(2N/\epsilon)}}\\
&= \frac{1}{4\sqrt{2\log(2N/\epsilon)}},
\end{align*}
and
\begin{align*}
\|\bfw'_i\|_{2} &\le \|\Phi_I\| \|(\Phi_I^T\Phi_I)^{-1}\|\|\Phi_I^T \phi_i\|_{2}
\le \sqrt{\frac32} \frac{2}{8\sqrt{2\log(2N/\epsilon)}}\\
& = \frac{\sqrt{3}}{8\sqrt{\log(2N/\epsilon)}}
\end{align*}
for all $i \in I^c.$
Let $a_1=1.5-\sqrt 2.$ Since $Z_{1,i}\sim \mathcal{N}(0,\|\bfw'_i\|_{2}^2)$, we have
\begin{align*}
\Pr(Z_1 > a_1\sqrt{2\log N})
&\le (N-k)\Pr\big(|Z_{1,i}| > a_1\sqrt{2\log N}\big)\\
& \le \frac{2(N-k)\|\bfw_i'\|_{2}}{a_1\sqrt{2\pi (2\log N)}}\; e^{-\frac{64}3a_1^2 \log N \log(2N/\epsilon)}\\
  &\le \frac{2.1}{\sqrt{(2\log N)\log(2N/\epsilon)}}N^{-0.15\log(2N/\epsilon)+1}\\
  &\le N^{-a}.
\end{align*}
(the multiplier in front of the exponent is less
than 1 for all $N>4$ and $\epsilon<1$).
Further, since the signs $\sgn(x_i),i\in I$ are uniform i.i.d. random variables,
we have
  \begin{align*}
   \Pr(Z_0>1/4)&\le (N-k)\Pr(|\ip{\bfw_i}{\sgn(\bfx_I)}|>1/4)\\
&\le 2(N-k)e^{-1/(32\|w_i\|_{2}^2)}\\
&<\epsilon.
\end{align*}
The proof is complete. \end{proof}

\vspace*{-.1in}Theorem \ref{thm:main2} is now easily established. Indeed, the assumptions
of Lemma \ref{lemma:128} are satisfied with probability at least $1-2\epsilon.$
The claim of the theorem follows from the above arguments.

\section{Sufficient conditions for statistical incoherence properties}
As discussed earlier, recovery properties of sampling matrices in linear programming decoding
procedures are controlled by the coherence parameter $\mu(\Phi)=\max_{i,j}\mu_{ij}.$ In particular, the Gershgorin theorem
implies that the condition $\mu=O(k^{-1})$ is sufficient for stable and robust recovery of signals
with sparsity $k$. In this section we show that this result can be improved to $\mu=O(k^{-1/4})$
in that the matrix satisfies the StRIP and SINC conditions.
The results of Sect. \ref{sect:lp-decoding} then imply stable recovery of generic random
$k$-sparse signals using linear programming decoding.

Let $\Phi$ be an $m\times N$ sampling matrix with columns $\phi_i,i=1,\dots,N.$
As above, let $\mu_{ij}=|\phi_i^T\phi_j|.$ 
Call the matrix $\Phi$ {\em coherence-invariant} the set $M_i:=\{\mu_{ij}, j\in [N]\backslash i\}$
is independent of $i$. Observe that most known constructions of sampling matrices satisfy
this property. This includes matrices constructed from linear codes \cite{dev07,bar10,nel11},
chirp matrices and various Reed-Muller matrices \cite{baj10a,cal10a}, as well as subsampled Fourier matrices \cite{hau10}. 
Our arguments change slightly if the matrix is not coherence-invariant. To deal simultaneously with both
cases, define the parameter $\theta=\theta(\Phi)$ as $\theta=\bar\mu^2$ if $\Phi$ is coherence-invariant and
$\theta=\bar\mu_{\max}^2$ otherwise.

The next theorem gives sufficient conditions for the SINC property in terms
of coherence parameters of $\Phi.$
%
\begin{theorem}\label{thm:sinc} Let $\Phi$ be an $m\times N$ matrix with unit-norm columns, coherence $\mu$
and square coherence $\theta.$ Suppose that $\Phi$ is coherence-invariant,
  \begin{equation}\label{eq:mu2a}
  \mu^4\le \frac{(1-a)^2\beta^2}{32 k(\log 2N/\epsilon)^{3}}\text{\quad and \quad} \theta\le \frac{a\beta}{k\log(2N/\epsilon)} ,
  \end{equation}
  where $\beta>0$ and $0<a<1$ are any constants.
  Then $\Phi$ has the $(k,\alpha,\epsilon)$-SINC property with $\alpha=\beta/\log(2N/\epsilon).$

\end{theorem}
Before proving this theorem we will introduce some notation. Fix $j\in[N]$ and let $I_j=\{i_1,i_2,\dots,i_k\}$
be a random $k$-subset such that $j\not\in I_j.$ The subsets $I_j$ are chosen from the set $[N-1]$ with
uniform distribution. Define random variables $Y_{j,l}=\mu^2_{j,i_l}, l=1,\dots,k$. Next define
a sequence of random variables $Z_{j,t},t=0,1,\dots,k,$ where
  $$
   Z_{j,0}=\avg_{I_j}\sum_{l=1}^kY_{j,l},\quad Z_{j,t}=\avg_{I_j} \Big(\sum_{l=1}^k Y_{j,l}\mid Y_{j,1},Y_{j,2},
   \dots,Y_{j,t}\Big), \;t=1,2,\dots,k.
 $$
From the assumption of coherence invariance, the variables $Z_{j,t}$ for different $j$ are stochastically equivalent.
Let
  $$
  Z_t=\avg_j Z_{j,t}=\avg_{R_k'}\Big(\sum_{l=1}^k Y_{j,l}\mid Y_{j,1},Y_{j,2},
   \dots,Y_{j,t}\Big), \quad t=1,\dots,k.
  $$
The random variables $Z_t$ are defined on the set of $(k+1)$-subsets of $[N]$ with probability distribution $P_{R_k'}$. We will show
that they form a Doob martingale. Begin with defining a sequence of $\sigma$-algebras $\cF_t,t=0,1,\dots,k,$
where $\cF_0=\{\emptyset,[N]\}$ and $\cF_t, t\ge 1$ is the smallest $\sigma$-algebra with respect to
which the variables $Y_{j,1},\dots,Y_{j,t}$ are measurable (thus, $\cF_t$ is formed of all subsets of $[N]$ of size $\le t+1$).
Clearly, $\cF_0\subset\cF_1\subset\dots\subset\cF_k$, and for each $t,$ $Z_t$ is a bounded random variable that
is measurable with respect to $\cF_t.$ Observe that 
  \begin{align}
  Z_0&=\avg_j Z_{j,0}=\avg_{R_k'}\sum_{l=1}^k \mu_{j,i_l}^2=\sum_{l=1}^k\avg_{R_k'} \mu_{j,i_l}^2=k\bar\mu^2
  \label{eq:ce}\\
  &\le k\bar\mu_{\max}^2, \label{eq:nce}
  \end{align}
where \eqref{eq:ce} assumes coherence invariance, and \eqref{eq:nce} is valid independently
of that assumption.
\begin{lemma}\label{lemma:bounded}
The sequence $(Z_t,\cF_t)_{t=0,1,\dots,k}\;$ forms a bounded-differences martingale, namely
  $
  \avg_{R_k'}(Z_t\mid Z_0,Z_1,\dots,Z_{t-1})=Z_{t-1}$
  and
  $$
  |Z_t-Z_{t-1}|\le 2\mu^2\Big(1+\frac k{N-k-2}\Big), \quad t=1,\dots,k.
  $$
\end{lemma}
\begin{proof} In the proof we write $\avg$ instead of $\avg_{R_k'}.$ We have
       \begin{align*}
Z_t &= \avg \Big( \sum_{l=1}^{k} Y_{j,l} \mid \cF_t \Big)
=  \sum_{l=1}^{t} Y_{j,l} + \avg \Big( \sum_{l=t+1}^{k} Y_{j,l} \mid \cF_t \Big)\\
&=  Z_{t-1} +Y_{j,t} + \avg \Big( \sum_{l=t+1}^{k} Y_{j,l} \mid \cF_t \Big) - \avg \Big( \sum_{l=t}^{k} Y_{j,l} \mid \cF_{t-1} \Big).
\end{align*}
Next,
  \begin{align*}
\avg (Z_t \mid Z_0,Z_1,\dots,Z_{t-1}) &= Z_{t-1} + \avg(Y_{j,t}\mid Z_0,Z_1,\dots,Z_{t-1}) +\avg\Big(\avg\Big(
\sum_{l=t+1}^{k} Y_{j,l} \mid \cF_t \Big)\mid Z_0,\dots,Z_{t-1}\Big)\\
&\hspace*{.2in}-\avg\Big( \avg \Big( \sum_{l=t}^{k} Y_{j,l} \mid
\cF_{t-1} \Big)\mid Z_0,\dots,Z_{t-1}\Big)\\
& = Z_{t-1} +
\avg\Big(Y_{j,t} \mid Z_0,\dots,Z_{t-1} \Big) \\
&\hspace*{.2in}+\avg\Big( \sum_{l=t+1}^{k} Y_{j,l} \mid Z_0,\dots,Z_{t-1}\Big)-\avg
\Big( \sum_{l=t}^{k} Y_{j,l} \mid Z_0,\dots,Z_{t-1} \Big)\\
&= Z_{t-1},
\end{align*}
which is what we claimed.

Next we prove a bound on the random variable $|Z_t-Z_{t-1}|$.  We have
  \begin{align*}
  |Z_t-Z_{t-1}|  &= \Big|\avg \Big( \sum_{l=1}^{k} Y_{j,l} \mid \cF_t
\Big) - \avg \Big( \sum_{l=1}^{k} Y_{j,l} \mid \cF_{t-1}
\Big)\Big|\\ 
&\le \max_{a,b} \Big|\avg \Big( \sum_{l=1}^{k} Y_{j,l} \mid
\cF_{t-1}, Y_{t,l}=a \Big) - \avg \Big( \sum_{l=1}^{k} Y_{j,l} \mid
\cF_{t-1}, Y_{t,l} =b \Big)\Big|\\ 
&= \max_{a,b}
\Big|\sum_{l=1}^{k}\Big(\avg \Big( Y_{j,l} \mid \cF_{t-1}, Y_{t,l}=a \Big) -
\avg \Big( Y_{j,l} \mid \cF_{t-1}, Y_{t,l} =b \Big)\Big)\Big|\\
&
= \max_{a,b}
\Big| a-b + \sum_{l=t+1}^{k}\Big(\avg \Big( Y_{j,l} \mid \cF_{t-1},
Y_{t,l}=a \Big) - \avg \Big( Y_{j,l} \mid \cF_{t-1}, Y_{t,l} =b
\Big)\Big)\Big|\\
&\le \Big| 2\mu^2 
+ \sum_{l=t+1}^{k}  \frac{2\mu^2}{N-l-2}\Big|\\
&= 2\mu^2 \frac{N-2}{N-k-2}\end{align*}
\end{proof}

\vspace*{-.1in}To prove Theorem \ref{thm:sinc} we use the Azuma-Hoeffding inequality (see, e.g., \cite{mcd89}).
\begin{proposition} \label{prop:AH}{\rm (Azuma-Hoeffding)} Let $X_0,\dots,X_{k-1}$ be a martingale
with $|X_{i}-X_{i-1}|\le a_i$ for each $i$, for suitable constants $a_i.$
Then for any $\nu>0,$
  $$
  \Pr\Big(\Big|\sum_{t=1}^{k-1} (X_i-X_{i-1})\Big|\ge \nu\Big)\le 2\exp \frac{-\nu^2}{2\sum a_i^2}.
  $$
\end{proposition}

{\em Proof of Theorem \ref{thm:sinc}:} Bounding large deviations for the 
sum $|\sum_{t=1}^{k}(Z_t-Z_{t-1})|=|Z_{k}-Z_0|,$
 we obtain
   \begin{equation}\label{eq:num}
   \Pr(|Z_{k} -Z_0|>  \nu)
        \le 2\exp\Big(-\frac{\nu^2}{8\mu^4k(\frac{N-2}{N-k-2})^2} \Big),
   \end{equation}
   where the probability is computed with respect to the choice of {\em ordered} $(k+1)$-tuples in $[N]$ and $\nu>0$ is any constant. Assume coherence invariance.
Using \eqref{eq:ce} and the inequality $(N-2)/(N-k-2)<2$ valid for all $k<\frac N2-1,$ we obtain
  $$
  \Pr(Z_k\ge \nu+k\bar\mu^2)\le\Pr(|Z_k-k\bar\mu^2|\ge \nu)\le 2\exp \Big(-\frac{\nu^2}{32\mu^2 k}\Big).
  $$
Now take $\beta>0$ and $\nu=\frac\beta{\log(2N/\epsilon)}-k\bar\mu^2.$
Suppose that for some $a\in(0,1)$
   \begin{equation}\label{eq:e1}
   k\mu^4 \le \frac{((1-a)\beta)^2}{32}\Big(\log\frac{2N}{\epsilon}\Big)^{-3}, \quad
      k\bar\mu^2\le \frac{a\beta}{\log(2N/\epsilon)}, 
   \end{equation}
then we obtain
    \begin{equation}\label{eq:e2}
  \Pr \Big(\|\Phi_{I_j}^T\phi_j\|_2^2\ge \frac\beta{\log(2N/\epsilon)}\Big)\le  2\exp \Big(-\frac{\nu^4}{32\mu^4 k}\Big)
  \le \frac\epsilon N
  \end{equation}
Now the first claim of Theorem \ref{thm:sinc} follows by the union bound with respect to the choice of the
index $j$. 

Assume that $\Phi$ does not satisfy the invariance condition. Then we rely on \eqref{eq:nce}
and repeat the above argument with respect to $\bar\mu_{\max}^2.$ \hfill\qed

The above proof contains the following statement.\begin{corollary}\label{cor:some}
Let $\Phi$ be an $m\times N$ matrix with coherence $\mu$ and $\theta=\bar\mu^2$ or $\bar\mu_{\max}^2,$
as appropriate.
Let $a\in(0,1)$ and $\beta>0$ be any constants. 
Suppose that for $\alpha<\beta \log_2e,$
  $$
    \mu^4\le \frac{(1-a)^2\alpha^3}{32\beta k},\quad k\theta\le a\alpha.
  $$
  Then $P_{R_k'}(\sum_{l=1}^k\mu_{i_l,j}^2\ge\alpha)\le2e^{-\beta/\alpha}.$
\end{corollary}
\begin{proof} Denote $\alpha=\beta/(\log(2N/\epsilon)),$ then $\epsilon/N=2e^{-\beta/\alpha}.$ The claim is obtained
by substituting $\alpha$ in \eqref{eq:e1}-\eqref{eq:e2}. \end{proof}

We note that this corollary follows directly from the SINC property under our assumptions on coherence and mean square coherence.
We observe that the SINC property naturally implies some StRIP condition as given in the following theorem.
\begin{theorem}\label{thm:strip} Let $\Phi$ be an $m\times N$ matrix. Let $I\subset[N]$ be
a random ordered $k$-subset and
suppose that for all $j\in I$, $\Pr(\sum_{m=1}^{k-1}\mu_{j,i_m}^2>\delta^2/k)<\epsilon_1/k.$ Then
$\Phi$ is a $(k,\delta,\epsilon_1)$-StRIP matrix.
\end{theorem}
\begin{proof} 
Given $I$ let $H(I)=\Phi_I^T\Phi_I-\text{Id}$ be the ``hollow Gram matrix". Let $B=\{I: \|H(I)\|_2>\delta\}\subset
\cP_k(N).$ We need to prove that $P_{R_k}(B)\le\epsilon.$
Let $(e_1,\dots,e_k)$ be the standard
basis of $\reals^k.$ Define a subset $C\subset\cP_k(N)$ as follows:
  $$
   C=\{I: \exists i \in I \text{ s.t. } \|H(I) e_i\|_2\ge \delta/\sqrt k\}
  $$
Let us show that $B\subseteq C$ by proving $C^c\subseteq B^c$. Indeed, if $I\in C^c$, then we have
  \begin{align*}
    \|H(I)\|&=\max_{|\bfx\|_2=1}\|H(I)\bfx\|_2
    =\max_{|\bfx\|_2=1}\|H(I)(x_1 e_1+x_2e_2+\dots +x_k e_k)\|\\
    &\le \max_{|\bfx\|_2=1} \sum_l |x_l|\, \|H(I)e_l\|_2\\
    &\le \max_{|\bfx\|_2=1} \|\bfx\|_1\max_{1\le l\le k} \|H(I) e_l\|_2 \\
    &\le \sqrt k \max_{1\le l\le k} \|H(I) e_l\|_2. \\
    &\le \delta,
  \end{align*}
which implies $I\in B^c.$ Now since $B\subseteq C$, we only need to show that $P_{R_k}(C)\le \epsilon$. 

Careful readers may have already noticed that the target quantity $P_{R_k}(C)$ uses a different probability measure from 
that in theorem's assumption. We note that a change of measure is actually inevitable since the probability 
measure in Azuma-Hoeffding's inequality we used in Proposition \ref{prop:AH} is with respect to ordered 
$k$-tuples while that in the definition of StRIP is with respect to unordered ones. In the following, we provide a rigorous 
calculation that supports this measure transformation.
 
For any $I\in C$, by definition, there exists at least one $l\in I$ such that $\|H_I e_l\|\ge \delta/\sqrt k$. Among such $l$, 
let $i(I)$ be the smallest one $i(I)=\min\{l\in I: \|H_I e_l\|_2\ge \delta/\sqrt k\}$. 
Now we define a map from an unordered $k$-tuple $I\in C\subseteq \cP_k(N)$ to a set of ordered 
$k$-tuples $Q(I)=\{(i_1,\dots,i_{k-1}, i(I)): (i_1,\dots,i_{k-1})=\sigma(I\backslash{i(I)}), \sigma\in S_{k-1}\},$
where $S_{k-1}$ denotes the set of all permutations of $k-1$ elements. 
Obviously, $|Q(I)|=(k-1)!$ for all $I$, and $Q(I_1)\cap Q(I_2)=\emptyset$
for distinct $k$-subsets $I_1,I_2.$ Moreover,  if $(i_1,\dots,i_k)\in Q(I)$, then $\|H(I) e_k\|_2\ge \delta/\sqrt k$
or $\sum_{l=1}^{k-1}\mu_{i_l,i_k}^2>\delta^2/k.$ Therefore 
  $$
     \bigcup_{I\in C} Q(I)\subseteq\big\{(i_1,\dots,i_k)\subset [N]: \sum_{l=1}^{k-1}\mu_{i_l,i_k}^2>\delta^2/k.\big\}
  $$
Now compute
  \begin{align*}
   P_{R_k}(B)&=\frac {|B|}{\binom Nk}\le \frac {|C|(k-1)!}{\binom Nk(k-1)!}=\frac{\sum_{I\in C}|Q(I)|}{\binom Nk(k-1)!}\\
     &=\frac{\big|\bigcup_{I\in C} Q(I)\big|}{\binom Nk(k-1)!}\\
     &\le \frac k{k!\binom Nk}\Big|\Big\{(i_1,\dots,i_k)\subset [N]: \sum_{l=1}^{k-1}\mu_{i_l,i_k}^2>\delta^2/k\Big\}\Big| \\
    &=k\Pr(\sum_{m=1}^{k-1}\mu_{j,i_m}^2>\delta^2/k).
  \end{align*}
  By the assumption of the theorem the last expression is at most $\epsilon$ which proves our claim.
\end{proof}

Theorem \ref{thm:strip} implies the following
\begin{corollary}\label{thm:strip1} Let $\Phi$ be an $m\times N$ matrix. If
  $$
  \theta\le \frac {a\delta^2}{k^2}, \quad\text{and}\quad \mu^4\le \frac{(1-a)^2\delta^4}{32 k^3\log(2k/\epsilon_1)},
  $$
where $0<a<1,$ then $\Phi$ is $(k,\delta,\epsilon_1)$-StRIP.
\end{corollary}
\begin{proof}
Take $\epsilon_1=2ke^{-\beta/\alpha},$ then $\beta=\frac{\delta^2}k\log(2k/\epsilon_1).$
The claim is obtained by substituting this value into the conditions of Corollary \ref{cor:some}. 
\end{proof}
Observe that the sufficient condition for the $(k,\delta)$-RIP property from the Gershgorin
theorem is $\mu<\delta/k,$ so the result of Corollary \ref{thm:strip1} gives a better result, namely $\mu=O(k^{-3/4}).$
At the same time, Tropp's result in \cite[Thm. B]{tro08} implies that the matrix 
$\Phi$ is $(k,\delta,\epsilon)$-StRIP under a weaker (i.e., more inclusive) condition.
Below we improve upon these results by analyzing the StRIP property directly rather than 
relying on the SINC condition. 

\begin{theorem}\label{thm:1/4} Let $\Phi$ by an $m\times N$ matrix and
let $\theta=\bar\mu^2$ or
$\theta=\bar\mu_{\max}^2,$ depending on whether $\Phi$ is coherence-invariant or not.
Let $\epsilon<\min\{1/k,e^{1-1/\log 2}\}$ and suppose that $\Phi$ satisfies
  \begin{equation}\label{eq:sc}
  k\mu^4\le \frac 1{\log^2(1/\epsilon)}\min\Big(\frac{(1-a)^2b^2}{32\log(2k)\log(e/\epsilon)},{c^2}\Big)\quad
\text{and}\quad k\theta\le\frac{ab}{\log(1/\epsilon)},
       \end{equation}
where $a,b,c\in(0,1)$ are constants such that
  \begin{equation}\label{eq:abc}
  \sqrt{a}+\sqrt{2ab}+\sqrt c+\frac {2k}N\|\Phi\|^2\le e^{-1/4}\delta/{6\sqrt 2}.
  \end{equation}
Then $\Phi$ is $(k,\delta,\epsilon)$-StRIP.
\end{theorem}
The proof relies on several results from \cite{tro08}. The following theorem is a modification
of Theorem 25 in that paper. Below $R$ denotes a linear operator that performs a restriction
to $k$ coordinates chosen according to some rule (e.g., randomly). Its domain is determined by the context.
Its adjoint $R^\ast$ acts on $\reals^k$ by padding the $k$-vector with the appropriate number of zeros.
\begin{theorem}\label{thm:dec} {\rm (Decoupling of the spectral norm)} Let $A$ be a $2N\times 2N$ symmetric matrix
with zero diagonal. Let $\eta\in\{0,1\}^{2N}$ be a random vector with $N$ components equal to one.
Define the index sets $T_1(\eta)=\{i:\eta_i=0\}, T_2(\eta)=\{i:\eta_i=1\}.$ Let $R$ be a random restriction to $k$ coordinates.
For any $q\ge 1$ we have
  \begin{equation}\label{eq:RR}
    (\avg\|RAR^\ast\|^q)^{1/q}\le 2\max_{k_1+k_2=k}\avg_{\eta}(\avg\|R_1A_{T_1(\eta)\times T_2(\eta)}R_2^\ast\|^q)
^{1/q},
  \end{equation}
where $A_{T_1(\eta)\times T_2(\eta)}$ denotes the submatrix of $A$ indexed by $T_1(\eta)\times T_2(\eta)$
and the matrices $R_i$ are independent restrictions to $k_i$ coordinates from $T_i,i=1,2.$

When $A$ has order $(2N+1)\times(2N+1),$ then an analogous result holds for partitions into blocks of
size $N$ and $N+1.$
\end{theorem}
 Inequality \eqref{eq:RR} is implicitly proved in the proof of the decoupling theorem (Theorem 9) \cite{tro08}. The ideas behind it are due to \cite{led91}.

The next lemma is due to Tropp \cite{tro08b} and Rudelson and Vershinin \cite{rud07}.
\begin{lemma}\label{lemma:q} Suppose that $A$ is a matrix with $N$ columns and let
$R$ be a random restriction to $k$ coordinates. Let $q\ge 2, p=\max(2,2\log(\rank AR^\ast),q/2).$ Then
   $$
   (\avg\|AR^\ast\|^q)^{1/q}\le 3\sqrt p(E\|AR^\ast\|^q_{1\to 2})^{1/q}+\sqrt{\frac kN}\|A\|
   $$
   where $\|\cdot\|_{1\to 2}$ is the maximum column norm.
\end{lemma}
The following lemma is a simple application of Markov's inequality, a similar result can be found in \cite{led91}, Lemma 4.10; see also \cite{tro08}.
\begin{lemma}\label{lemma:LD} Let $q,\lambda>0$ and let $\xi_q$ be a positive function of $q$.
Suppose that $Z$ is a positive random variable whose $q$th moment satisfies the bound
  $$
  (\avg Z^q)^{1/q}\le \xi_q \sqrt q+\lambda.
  $$
Then
  $$
  P(Z\ge e^{1/4}(\xi_q \sqrt q +\lambda))\le e^{-q/4}.
  $$
\end{lemma}
{\em Proof:} 
By the Markov inequality,
  $$
  P\left(Z\ge e^{1/4}(\xi_q \sqrt q+\lambda)\right)\leq \frac{\avg Z^q}{(e^{1/4}(\xi_q\sqrt q+\lambda))^q}\leq \left(\frac{\xi_q\sqrt{q}+\lambda}{e^{1/4}(\xi_q\sqrt q+\lambda)}\right)^q=e^{-q/4}. \hspace*{1in}\qed
  $$

The main part of the proof is contained in the following lemma.
\begin{lemma}\label{lemma:technical}
Let $\Phi$ be an $m\times N$ matrix with coherence parameter $\mu.$ Suppose that for some $0< \epsilon_1,\epsilon_2<1$
  \begin{equation}\label{eq:asp}
   P_{R_k'}(\{(I,i):\|\Phi_I^T\phi_i\|^2\ge \epsilon_1\}\mid i)\le \epsilon_2.
   \end{equation}
Let $R$ be a random restriction to $k$ coordinates and $H=\Phi^T\Phi-\text{Id}.$ For any $q\ge 2, p=\max(2,2\log(\rank RHR^\ast),q/2)$ we have
  \begin{equation}\label{eq:RHR}
  (\avg\|R H R^\ast\|^q)^{1/q}\le 6\sqrt p (\sqrt \epsilon_1 +(k \epsilon_2)^{1/q} \mu \sqrt{k}
     +\sqrt {2k\theta}\,)+\frac {2k}N\|\Phi\|^2.
  \end{equation}
\end{lemma}
\begin{proof} We begin with setting the stage to apply Theorem \ref{thm:dec}. Let $\eta\in\{0,1\}^{N}$ be a random vector with $N/2$ ones and 
let $R_1,R_2$ be random restrictions to $k_i$ coordinates in the sets $T_i(\eta),i=1,2$, respectively.
Denote by $\supp(R_i),i=1,2$ the set of indices selected by $R_i$ and let $H(\eta):=H_{T_1(\eta)\times T_2(\eta)}$.
Let $q\ge 1$ and let us bound the term $\avg_\eta(\avg\|R_1H(\eta)R_2\|^q)^{1/q}$ that appears on the right side of \eqref{eq:RR}.
The expectation in the $q$-norm is computed for two random restrictions $R_1$ and $R_2$ that are conditionally independent
given $\eta.$ Let $\avg_i$ be the expectation with respect to $R_i,i=1,2$. Given $\eta$ we can evaluate these
expectations in succession and apply Lemma \ref{lemma:q} to $\avg_2:$
  \begin{align*}
  \avg_\eta(\avg \|R_1H(\eta)R_2^\ast\|^q)^{1/q}&=
  \avg_\eta\Big[\avg_1(\avg_2\|R_1H(\eta)R_2^\ast\|^q)^{q/q}\Big]^{1/q}\\
 &\le \avg_\eta\Big\{\avg_1\Big[ 3\sqrt p\, (\avg_2\|R_1 H(\eta)R_2^\ast\|_{1\to 2}^q)^{1/q}+\sqrt{\frac {2k_2}{N} }
 \|R_1H(\eta)\|\Big]^q\Big\}^{1/q}\\
 &\le \avg_\eta\Big\{3\sqrt p \;
 \Big[\avg_1\big(\avg_2\|R_1 H(\eta)R_2^\ast\|_{1\to 2}^q)\Big]^{1/q}+ \sqrt{\frac {2k_2}{N} }
 \Big[\avg_1\|R_1H(\eta)\|^q\Big]^{1/q}\Big\}
  \end{align*}
where on the last line we used the Minkowski inequality (recall that the random variables involved are finite). 
Now use Lemma \ref{lemma:q} again to obtain
 \begin{align} 
   \avg_\eta(\avg \|R_1H(\eta)R_2^\ast\|^q)^{1/q}&\le 3\sqrt p\, \avg_\eta\Big[\avg_1\avg_2
   \|R_1H(\eta)R_2^\ast\|_{1\to 2}^q\Big]^{1/q}
   +3\sqrt{\frac{2k_2p}N}\avg_\eta\big(\avg_1\| H(\eta)^\ast R_1^\ast\|_{1\to2}^q\big)^{1/q}\label{eq:terms}\\
   &+\sqrt{\frac{4k_1k_2}{N^2}}\avg_\eta\|H(\eta)^\ast\|.\notag
  \end{align}
  Let us examine the three terms on the right-hand side of the last expression.
Let $\eta(R_2)$ be the random vector conditional on the choice of $k_2$ coordinates. The sample
space for $\eta(R_2)$ is formed of all the vectors $\eta\in\{0,1\}^{N}$ such that $\supp(R_2)\subset T_2(\eta).$
In other words, this is a subset of the sample space $\{0,1\}^N$ that is compatible with a given $R_2.$
The random restriction $R_1$ is still chosen out of $T_1(\eta)$ independently of $R_2.$
Denote by $\tilde R$ a random restriction to $k_1$ indices in the set $(\supp(R_2))^c$ and let $\tilde\avg$
be the expectation computed with respect to it. We can write
  \begin{align*}
   \avg_\eta(\avg_1\avg_2\|R_1H(\eta)R_2^\ast\|_{1\to 2}^q)^{1/q}
    &\le (\avg_\eta\avg_1\avg_2\|R_1H(\eta)R_2^\ast\|_{1\to 2}^q)^{1/q}\\
       & =(\avg_2\tilde\avg
   \|\tilde R H(\eta)R_2^\ast\|_{1\to 2}^q)^{1/q}
   \end{align*}
Recall that $H_{ij}=\mu_{ij}{\bf1}_{\{i\ne j\}}$ and that $\tilde R$ and $R_2$ are $0$-$1$ matrices.
Using this in the last equation, we obtain
  \begin{equation}\label{eq:sum}
  \avg_2\tilde\avg\|\tilde RH(\eta)R_2^\ast\|_{1\to 2}^q\le \avg_2\tilde\avg
  \max_{j\in\supp(R_2)}\textstyle{\Big(\sum_{i\in\supp(\tilde R)}\mu_{ij}^2\Big)^{q/2}}.
  \end{equation}
Now let us invoke assumption \eqref{eq:asp}. Recalling that $k_1<k,$ we have
  $$
  P_{R_2,\tilde R}\Big( \textstyle {\max\limits_{j\in\supp(R_2)}\sum_{i\in\supp(\tilde R)}\mu_{ij}^2}\ge \epsilon_1\Big)
  \le k_2\epsilon_2.
  $$
Thus with probability $1-k_2\epsilon_2$ the sum in \eqref{eq:sum} is bounded above by $\epsilon_1.$ For the
other instances we use the trivial bound $k_1\mu^2.$ We obtain
  \begin{align*}
 3\sqrt p\, \avg_\eta\avg_1(\avg_2\|R_1H(\eta)R_2^\ast\|_{1\to 2}^q)^{1/q}&\le 3\sqrt p ((1-k_2\epsilon_2)\epsilon_1^{q/2}
 +k_2\epsilon_2(k_1\mu^2)^{q/2})^{1/q}\\
 &\le 3\sqrt p (\epsilon_1^{q/2}+k_2\epsilon_2(k_1\mu^2)^{q/2})^{1/q}\\
 &\le 3\sqrt p (\sqrt{\epsilon_1}+(k\epsilon_2)^{1/q} \sqrt{k_1\mu^2}),
 \end{align*} 
where in the last step we used the inequality $a^q+b^q
\le (a+b)^q$ valid for all $q\ge 1$ and positive $a,b.$
Let us turn to the second term on the right-hand side of \eqref{eq:terms}. Assuming coherence invariance, we observe that
  $$
  \|H(\eta)^\ast R_1^\ast\|_{1\to2}=\max_{j\in T_1(\eta)}\|H_{j,T_2(\eta)}\|_2\le \max_{j\in[N]}\|H_{j,\cdot}\|_2=\sqrt{N\bar\mu^2}
  $$
  where $H_{j,\cdot}$ denotes the $j$th row of $H$ and $H_{j,T_2(\eta)}$ is a restriction of the $j$th row to the 
  indices in $T_2(\eta).$ At the same time, if the dictionary is not coherence-invariant,
then in the last step we estimate the maximum norm from above by $\sqrt{N\bar\mu_{\max}^2},$
so overall the second term is not greater than $\sqrt{N\theta},$
  
Finally, the third term in \eqref{eq:terms} can be bounded as follows:
    \begin{align*}
    \sqrt{\frac{4k_1k_2}{N^2}}\avg_\eta\|H(\eta)\|&\le \sqrt{\frac{(k_1+k_2)^2}{N^2}} \|H\|=\frac kN \|\Phi^T\Phi-I_N\|\\
    &\le \frac kN\max(1,\|\Phi\|^2-1)\le \frac kN\|\Phi\|^2,
    \end{align*}
    where the last step uses the fact that the columns of $\Phi$ have unit norm, and so
    $\Phi^2\ge N/m>1.$
    
    Combining all the information accumulated up to this point in \eqref{eq:terms}, we obtain
    $$
   \avg_\eta(\avg \|R_1H(\eta)R_2^\ast\|^q)^{1/q}\le 3\sqrt p(\sqrt{\epsilon_1}+(k\epsilon_2)^{1/q}\mu\sqrt k+
   \sqrt{2k_2\theta}\,)+\frac kN\|\Phi\|^2.
   $$
   Finally, use this estimate in \eqref{eq:RR} to obtain the claim of the lemma.
\end{proof}

{\em Proof of Theorem \ref{thm:1/4}:} 
\begin{proof}
The strategy is to fix a triple $a,b,c\in (0,1)$ that satisfies \eqref{eq:abc} and to prove that \eqref{eq:sc} 
implies $(k,\delta,\epsilon)$-StRIP. 
Let $\epsilon_1=\frac{b}{\log1/\epsilon}$ and $\epsilon_2=k^{-1+\log\epsilon}$. In Corollary \ref{cor:some} set 
$\alpha=\epsilon_1$ and $\beta=\alpha\log(2/\epsilon_2).$  Under the assumptions in \eqref{eq:sc} 
this corollary implies that
   $$
    P_{R'}\Big(\sum\limits_{m=1}^k \mu_{i_m,j}^2>\epsilon_1\Big)<\epsilon_2.
  $$
Invoking Lemma \ref{lemma:technical}, we conclude that \eqref{eq:RHR} holds with the current values of $\epsilon_1,\epsilon_2$.
For any $q\geq 4\log k$ we have $p=q/2$, and thus \eqref{eq:RHR} becomes
 \begin{equation}\label{eq:ld}
  (\avg\|R H R^\ast\|^q)^{1/q}\le {3}\sqrt {2q} (\sqrt \epsilon_1 +(k \epsilon_2)^{1/q} \mu \sqrt{k}
     +\sqrt {2k\theta})+2\frac kN\|\Phi\|^2.
   \end{equation}
Introduce the following quantities:       
      $$\xi_q=3\sqrt{2}(\sqrt {\epsilon_1}+(k \epsilon_2)^{1/q} \mu \sqrt{k}
     +\sqrt {2k\theta}) \ \  \text{and} \ \  \lambda=\frac {2k}N\|\Phi\|^2.
   $$   
Now \eqref{eq:ld} matches the assumption of Lemma \ref{lemma:LD}, and we obtain
    \begin{equation}\label{eq:P_RHR}
      P_{R_k}(\|R H R^\ast\| \geq e^{1/4}(\xi_q\sqrt q +\lambda))\leq e^{-q/4}.
    \end{equation}
Choose $q=4\log(1/\epsilon),$ which is consistent with our earlier assumptions on $k,q,$ and $\epsilon.$
With this, we obtain
      $$P_{R_k}\big(\|R H R^\ast\|\geq e^{1/4}(\xi_q\sqrt q +\lambda)\big)\leq \epsilon.$$
Now observe that $\|R H R^\ast\|\leq \delta$ is precisely the RIP property for the support identified
by the matrix $R.$ Let us verify that the inequality 
   $$ 6\sqrt{2} \big(\sqrt \epsilon_1+(k\epsilon_2)^{1/q}\sqrt{k\mu^2}+\sqrt{2k\theta}\big)\sqrt{\log(1/\epsilon)}+\frac{2k}{N}\|\Phi\|^2<e^{-1/4}\delta$$
   is equivalent to \eqref{eq:abc}. This is shown by substituting $\epsilon_1$ and $\epsilon_2$ with their definitions, 
   and $\mu$ and $\theta$ with their bounds in statement of the theorem.
Thus, $P_{R_k}(\|R H R^\ast\|\geq \delta)\le\epsilon,$ which establishes the StRIP property of $\Phi.$
\end{proof}
\vspace*{.1in}

\remove{To see that matrices that satisfy the constraints of Theorem \ref{thm:sinc} exist,
take again the Delsarte-Goethals matrices \eqref{eq:example} with r=1.
Then $\mu=2m^{-1/2}$, so taking $k=\sqrt m=(2N)^{1/6}$ it is possible to satisfy the
constraints on $\mu$ and $m$ in \eqref{eq:constraints}.
In the next section we will see that it is possible to construct a broad class of
sampling matrices whose
parameters satisfy the assumptions of both Theorems \ref{thm:strip} and \ref{thm:sinc}.}

\section{Examples and extensions}\label{sec:determin}
\subsection{Examples of sampling matrices.} It is known \cite{don09} that experimental performance of many known RIP sampling matrices in sparse recovery is far better than
predicted by the theoretical estimates. Theorems \ref{thm:sinc} and \ref{thm:1/4} provide some insight into the reasons for such behavior. 
As an example, take binary matrices constructed from the Delsarte-Goethals codes 
\cite[p.461]{mac91}.
The parameters of the matrices are as follows:
  \begin{equation}\label{eq:example}
  m=2^{2s+2}, \;N=2^{-r}m^{r+2},\;\mu=2^rm^{-\half}
  \end{equation}
where $s\ge 0$ is any integer, and where for a fixed $s$, the parameter $r$ can be any number
in $\{0,1,\dots,s-1\}.$
If we take $s$ to be an odd integer and set $r=(s+1)/2$, then we obtain,
   $$
m= 2^{4r},\; N=2^{4r^2+7r},\; \mu=m^{-1/4}.
   $$
The matrix $\Phi$ is coherence-invariant, so we put $\theta=\bar\mu^2.$
Lemma \ref{lem:pless} below implies that
   \begin{equation}\label{eq:mu2}
\bar\mu^2=\frac{N-m}{m(N-1)}<\frac{1}{m},
   \end{equation}
and the norm of the sampling matrix satisfies $\|\Phi\|=\sqrt{N/m}$.
Thus for $\mu$ and $\bar\mu^2$ to satisfy the assumptions in 
Theorems \ref{thm:sinc} and \ref{thm:1/4}, we only need $m$, $ N$, and $k$ 
to satisfy the relation $m=\Theta(k \log^3 \frac{N}{\epsilon})$ which is nearly optimal. 
Similar logic leads to derivations of such relations for other matrices. We summarize 
these arguments in the next proposition, which shows that matrices with nearly optimal
sketch length support high-probability recovery of sparse signals chosen from the 
generic signal model.
\begin{proposition}\label{prop:final}
Let $\Phi$ be an $m\times N$ sampling matrix.  Suppose that it has coherence parameter
$\mu=O(m^{-1/4})$ and $\theta=O(m^{-1}),$ where $\theta=\bar\mu^2$ or $\theta=\bar\mu_{\max}^2$
according as $\Phi$ is coherence-invariant or not, and
  $$
\|\Phi\|=O(\sqrt{N/k}).
  $$
 If $m=\Theta(k(\log (N/\epsilon))^3),$ 
then $\Phi$ supports sparse recovery under
Basis Pursuit for all but an $\epsilon$ proportion of $k$-sparse signals chosen from the generic
random signal model $\cS_k$
\end{proposition}
We remark that the conditions on (mean or maximum) square coherence are generally easy to achieve. 
As seen from Table \ref{table} below, they are satisfied by most examples considered in the existing
literature, including both random and deterministic constructions. 
 The most problematic quantity is the coherence parameter $\mu$. 
It might either be large itself, or have a large theoretical bound. 
Compared to earlier work, our results rely on a more relaxed condition on $\mu$, 
enabling us to establish near-optimality for new classes of matrices.  
For readers' convenience, we summarize in Table 1 a list of such optimal matrices 
along with several of their useful properties. A systematic description of all but the last 
two classes of matrices can be found in \cite{baj11}. Therefore 
we limit ourselves to giving definitions and performing some not immediately obvious calculations 
of the newly defined parameter, the mean square coherence.

\vspace*{.1in}\emph{Normalized Gaussian Frames.} A normalized Gaussian frame is obtained
by normalizing each column of a Gaussian matrix with independent, Gaussian-distributed entries 
that have zero mean and unit variance. 
The mutual coherence and spectral norm of such matrices were characterized in \cite{baj11} (see Table \ref{table}). 
These results together with the relation $\bar\mu_{\max}^2<\mu^2$ lead to a trivial upper bound on 
$\bar\mu_{\max}^2$, namely $\bar\mu_{\max}^2\leq 15\log N/m$. 
Since this bound is already tight enough for $\bar\mu_{\max}^2$ to satisfy the assumption 
of Proposition \ref{prop:final}, and to avoid distraction from the main goals of the paper, 
we made no attempt to refine it here. 

\vspace*{.1in}\emph{Random Harmonic Frames}: Let $\mathcal{F}$ be an $N\times N$ discrete Fourier transform matrix, i.e., 
$\mathcal{F}_{j,k}=\frac{1}{\sqrt{N}}e^{2\pi i jk/N}$. Let $\eta_i$, $i=1,...,N$, be a sequence of independent Bernoulli random variables with mean $\frac{m}{N}$. Set $\mathcal{M}=\{i: \eta_i=1\}$ and use $\mathcal{F}_{\mathcal{M}}$ to denote the submatrix of $\mathcal{F}$ whose row indices lies in $\mathcal{M}$. Then the random matrix $\sqrt{\frac{|\mathcal{M}|}{N}}\mathcal{F}_{\mathcal{M}}$ is called a random harmonic frame
\cite{can06a,can06b}. In the
next proposition we compute the mean square coherence for all realizations of this matrix.
\begin{proposition} All instances of the random harmonic frames are coherence invariant with the following mean square coherence
\[
\bar\mu^2 =\frac{N-|\mathcal{M}|}{(N-1)|\mathcal{M}|}.
\]
\end{proposition}
{\em Proof:}
For each $t\in [|\mathcal{M}|]$, let $a_t$ with  be the $t$-th member of $\mathcal{M}$.
To prove coherence invariance, we only need to show that 
$\{\mu_{j,k}: k\in [N]\backslash j\}=\{\mu_{N,k}: k\in [N-1]\}$ holds for all $j\in [ N ]$. This is true since
\[
\mu_{j,k}=\frac{1}{|\mathcal{M}}\sum\limits_{t=1}^{|\mathcal{M}|} e^{\frac{2\pi i (j-k) a_t}{N}}=\mu_{N,(k-j+N)\text{mod } N}
\quad \text{for all } k\neq j.
\]
In words, the $k$th coherence in the set $\{\mu_{j,k}, k\in [N]\backslash j\}$ is exactly the $\left(k-j+N\mod  N\right)$-th coherence in $\{\mu_{N,k}, k\in [N-1]\}$, therefore the two sets are equal. We proceed to calculate the mean square coherence,
\begin{align*}
\bar\mu^2&=\frac{1}{N(N-1)|\mathcal{M}|^2}\sum\limits_{j\neq k,j,k=1}^N\left|\sum\limits_{t=1}^{|\mathcal{M}|} e^{2\pi i (j-k)a_t/N}\right|^2 \\
&=\frac{1}{N(N-1)|\mathcal{M}|^2}\sum\limits_{j\neq k, j,k=1}^N \sum\limits_{t_1,t_2=1}^{|\mathcal{M}|} e^{2\pi i(j-k)(a_{t_1}-a_{t_2})/N}\\
&=\frac{1}{N(N-1)|\mathcal{M}|^2}\left(\sum\limits_{j\neq k,j,k=1}^N \sum \limits_{t_1=t_2=1}^{|\mathcal{M}|} 1+\sum\limits_{t_1\neq t_2 ,t_1,t_2=1}^{|\mathcal{M}|}\sum\limits_{k=1}^N\sum\limits_{j\neq k} e^{2\pi i(j-k)(a_{t_1}-a_{t_2})/N}\right)\\
&=\frac{1}{N(N-1)|\mathcal{M}|^2}(N(N-1)|\mathcal{M}|-|\mathcal{M}|(|\mathcal{M}|-1)N)\\
&=\frac{N-|\mathcal{M}|}{(N-1)|\mathcal{M}|}. \hspace*{4in}\text{\qed}
\end{align*}

\vspace*{.1in}\emph{Chirp Matrices}: Let $m$ be a prime. An $m\times m^2$ ``chirp matrix'' $\Phi$ is defined by $\Phi_{t,am+b}=\frac{1}{\sqrt{m}}e^{2\pi i (bt^2+at)/m}$ for $t,a,b=1,...,m$. The coherence between each pairs of column vectors is known to be 
    $$
\mu_{jk}=\frac{1}{\sqrt{m}}     \quad(j\ne k),
   $$
 from which we immediately obtain the inequalities $\mu \leq 1/\sqrt{m}$ and $\bar\mu^2 \leq 1/m$. More details on these
 frames are given, e.g., in \cite{Brodzik06,Casazza06}.

\vspace*{.1in}\emph{Equiangular tight frames (ETFs)}: A matrix $\Phi$ is called an ETF if its 
columns $\{\phi_i \in \mathbb{R}^m, i=1,...,N\}$ satisfy the following two conditions:
\begin{itemize}
\item $\|\phi_i\|_2=1$, for $i=1,...,N$.
\item $\mu_{ij}=\sqrt{\frac{N-m}{m(N-1)}}$, for $i\neq j$.
\end{itemize}
From this definition we obtain $\mu=\sqrt{\frac{N-m}{m(N-1)}}$ and $\theta=\bar\mu^2=\frac{N-m}{m(N-1)}$.
The entry in the table also covers the recent construction of ETFs from Steiner systems \cite{Fickus12}.

\vspace*{.1in}\emph{Reed-Muller matrices:} In Table \ref{table} we list two tight frames
obtained from binary codes. The Reed-Muller matrices are obtained from certain special subcodes
of the second-order Reed-Muller codes \cite{mac91}; their coherence parameter $\mu$ is found in \cite{baj11}
and the mean square coherence is found from \eqref{eq:mu2}. The Delsarte-Goethals matrices
are also based on some subcodes of the second order Reed-Muller codes and were discussed earlier in 
this section. 
Both dictionaries support orthogonal arrays, and therefore, form unit-norm tight frames (rows
of the matrix $\Phi$ are pairwise orthogonal), with a consequence that $\|\Phi\|=\sqrt{N/m}.$
We include these two examples out of many other possibilities based on codes because they
appear in earlier works, and because their parameters are in the range 
that fits well our conditions.

We note that the quaternary version of these frames is also of interest in the context
of sparse recovery; see in particular \cite{cal10a}.

\vspace*{.05in}
\emph{Deterministic Fourier Construction\cite{hau10}}: Let $p>2$ be a prime, 
and let $f(x)\in \ff_p[x]$ be a polynomial of degree $d>2$ over the finite field $\ff_p$. 
Suppose that $m$ is some integer satisfying 
$p^{1/(d-1)}\leq m\leq p$. 
Then we can construct an $m\times p$ deterministic RIP matrix from a $p\times p$ DFT matrix 
by keeping only the rows with indices in $\{f(n)\!\! \pmod p, n=1,\dots,m\},$ 
and normalizing the columns of the resulting matrix.
It is known that this matrix has mutual coherence no greater than $e^{3d}m^{-1/(9d^2\log d)}$. 
Even though this bound is an artifact of the proof technique used in \cite{hau10}, 
there seem to be no obvious ways of improving it. 

\vspace*{.1in}{\footnotesize

\begin{table}[t]\begin{center}
     \begin{tabular}{| l | c  c  c c c |}
     \hline
Name & $\mathbb{R}$/$\mathbb{C}$ & Coherence-Invariant  & Dimensions & $\mu$  & $\theta(\Phi)$   \\[.05in]
   \hline
     Normalized Gaussian (G) & $\mathbb{R}$  &  No & $m\times N$ & $\leq \frac{\sqrt{15\log N}}{\sqrt{m}-\sqrt{12\log N}}$ &  $\leq \mu^2$
        \\[.1in]
     Random harmonic  (RH)   & $\mathbb{C}$  & Yes & $|\mathcal{M}|\times N$, $\frac{1}{2}m\leq |\mathcal{M}|\leq \frac{3}{2}m$   & $\leq \sqrt{\frac{118(N-m)\log N}{mN}}$  &      $\leq\frac{N-|\mathcal{M}|}{|\mathcal{M}|(N-1)}$ \\[.1in]
     Chirp  (C)     & $\mathbb{C}$ & Yes &  $m\times m^2$ & $\frac{1}{\sqrt m}$      &      $\frac{1}{m+1} $   \\[.05in]
     ETF (including Steiner) &  $\mathbb{C}$      & Yes &  $\sqrt{N}\leq m\leq N$  & $\sqrt{\frac{N-m}{m(N-1)}}$ &  $\mu^2$  \\[.05in]
     Reed-Muller (RM) & $\mathbb{R}$        & Yes &  $2^s\times 2^{t(1+s)}$ & $\leq \frac{1}{\sqrt{2^{s-2t-1}}}$ ,    &  $\leq 2^{-s} $      \\[.05in]
  Delsarte-Goethals set (DG)    &$\mathbb{R} $      & Yes &  $2^{2s+2}\times 2^{2(s+1)(r+2)-r}$    &$2^{r-s-1} $ &     $ \leq 2^{-2s-2}$   \\[.05in]
   Deterministic subFourier (SF) & $\mathbb{C}$        & Yes & $ m\times p$ & $e^{3d}m^{-1/(9d^2\log d)}$        & $\leq \frac{1}{m}$   \\[.05in]
   \hline
     \end{tabular}
\end{center}
\vspace*{.1in}
\begin{center}
   \begin{tabular}{| l |c c c c  |}
   \hline
    Name &      $\|\Phi\|$  &Restrictions &  Probability       & Sketch dimensions: $m= \Omega(\cdot)$     \\[.05in]
   \hline
    G   & $\leq \frac{\sqrt{m}+\sqrt{ N}+\sqrt{2\log N}}{\sqrt{m-\sqrt{8m\log N}}}$    &  
            $60\log N\leq m\leq \frac{N-1}{4\log N}$      &      $\geq 1-\frac{11}{N}$ &$k\log^2 N$ \\[.1in]
    RH       & $\leq \frac{2N}{m}$   & $16\log N\leq m\leq \frac{N}{3}$ & $\geq 1-\frac{4}{N}-\frac{1}{N^2}$  &$ k\log^2 N$  \\[.05in]    
   C & $\sqrt m$  &  $m$ is prime &      deterministic &$ k\log m$  \\[.05in]
    ETF &  $\sqrt{\frac{N}{m}}$ & $\sqrt{\frac{M(N-1)}{N-M}}$,$\sqrt{\frac{(N-m)(N-1)}{m}}$ are odd integers  &  deterministic  &$ k\log m$ \\[.05in] 
    RM & $2^{ts/2}$ & $t<s/4$        &     deterministic &$k\log^3 N$ \\[.05in]
   DG        & $2^{(s+1)(r+1)-r/2}$       & $r<s/2$     & deterministic  &$k\log^3 N$  \\[.05in]
  SF  &$\sqrt{p/m}$  & $p$ is prime, $p^{1/(d-1)}\leq m\leq p$      &  deterministic  &$(k\log^2 N)^{\frac{9d^2\log d}{4}}$  \\[.05in] 
   \hline
    \end{tabular}
\end{center}\caption{Examples for Prop.~\ref{prop:final}: Classes of sampling matrices satisfying the incoherence conditions}\label{table}
\end{table}}

\vspace*{-.05in}
\subsection{StRIP matrices from orthogonal arrays} Let us briefly consider another way
of constructing StRIP matrices based on elementary arguments.
Let $\cC=\{\phi_1,\dots,\phi_N\}$ be a collection of binary $m$-vectors.
We assume that the entries of the vectors are of the form $\pm 1/\sqrt m$ and
denote the correlation of $\phi_i$ and $\phi_j$ by $\mu_{ij}=|\ip{\phi_i}{\phi_j}|.$

The set $\cC$ is called an orthogonal array of strength $t$ if every subset
of $r\le t$ coordinates of the vectors of $\cC$ supports a uniformly random binary $r$-vector.
A good reference for orthogonal arrays is the book
by Hedayat et al. \cite{hed99}. An orthogonal array has the property that any $t$
coordinates of a randomly chosen vector behave as independent random variables
(therefore, of course, $t$ is much smaller than $m$).
In particular, the first $t$
moments of the distance distribution of $\cC$ are equal to the moments of the
binomial distribution. Let $d_{ij}=\frac m2(1-\phi_i^T\phi_j)$ be the Hamming distance
between $\phi_i$ and $\phi_j.$
\begin{lemma} {\rm (Pless identities, e.g. \cite[p.132]{mac91})}
\label{lem:pless}
Let $\cC$ be an orthogonal array of strength $t$.  Let $B_w=(1/N)|\{
(\phi_i,\phi_j)\in\cC^2\mid d_{ij}=w\}|$ be the number of pairs vectors in
$\cC$ at distance $w$.  For all $l=1,2,\dots,t$
   \begin{equation}\label{eq:ppm}
  \sum_{w=0}^m \frac{B_w}{N}\Big(w-\frac m2\Big)^l=\frac1{2^m}\sum_{w=0}^m\binom mw
       \Big(w-\frac m2\Big)^l.
  \end{equation}
\end{lemma}

We will need a manageable estimate of the right-hand side of \eqref{eq:ppm}.
We quote from \cite[p.288]{mac91}: let $l\ge 2$ be even, then
  \begin{equation}\label{eq:ecm}
  \frac1{2^m}\sum_{w=0}^m\binom mw\Big(w-\frac m2\Big)^l\le
    \Big(\frac {ml}{4e}\Big)^{l/2}\sqrt l\,e^{1/6}.
   \end{equation}
The main result of this section is given by the following theorem.
\begin{theorem}
Let $\cC$ be an orthogonal array of strength $t$ and cardinality $N$ and
let $l\le t$ be even.
If $m\ge (3/4)\, l(k/\delta)^2(k/\epsilon)^{2/l}$ then $\Phi$ is $(k,\delta,\epsilon)$-StRIP.
\end{theorem}
\begin{proof}
  Let $I\subset [N]$ be a uniformly random $k$-subset. We clearly have
  $$
  \lambda_{\min}(\Phi^T_I\Phi_I)\|\bfx\|_{2}^2\le \|\Phi_I\bfx\|_{2}^2\le\lambda_{\max}(\Phi^T_I\Phi_I)\|\bfx\|_{2}^2,
  $$
  where $\lambda_{\rm min}(\cdot)$ and  $\lambda_{\rm max}(\cdot)$ are the minimum and
maximum eigenvalues of the argument.

By the Gershgorin theorem, any eigenvalue $\lambda$ of the Gram matrix $\Phi_I^T\Phi$
satisfies
  $$
    |\lambda-1|\le \sum_{j\in I_i}\mu_{ij},
 $$
for some $i \in [N],$ where we used the notation $I_i:=I\backslash\{i\}.$ Now consider the probability
that for some $i\in I$ the sum  $\sum_{j\in I_i}\mu_{ij} >\delta.$
The proof will be finished if we show that this probability is less than $\epsilon.$
Let $I=\{i_1,\dots,i_k\}$. We have
  \begin{align*}P_{R_k}\Big(\exists i\in I: \sum_{j\in I_i}\mu_{ij}> \delta\Big)
&\le k P_{R_k} \Big(\sum_{j\in I_{i_1}}\mu_{i_1,j}> \delta\Big)
\le k  \frac{1}{\delta^l}
 \avg_{R_k}\Big(\sum_{j\in I_{i_1}}\mu_{i_1,j}\Big)^l\\
 &= k\frac{(k-1)^l}{\delta^l}
\avg_{R_k}\Big(\frac1{k-1}\sum_{j\in I_{i_1}}\mu_{i_1,j}\Big)^l\\
&\le  \frac{k(k-1)^{l-1}}{\delta^l} \avg_{R_k}\sum_{j\in I_{i_1}}\mu_{i_1,j}^l,
\end{align*}
where the last step uses convexity of the function $z\mapsto z^l.$
The trick is to show that the expectation on the last line,
presently computed over the choice of $I$, can be also found with respect
to a pair of random uniform elements of $\cC$ chosen without replacement. This is
established in the next calculation: 
  \begin{align}\label{eq:Ei}
 \avg_{R_k}\sum_{j\in I_{i_1}} \mu_{i_1,j}^l&=
\sum_{i_1<i_2<\dots<i_k}\frac1{\binom Nk}\sum_{j=2}^k
\mu_{i_1,i_j}^l=\frac1{k!\binom Nk}\sum_{i_1\ne i_2\ne\dots\ne i_k} \sum_{j=2}^k
\mu_{i_1,i_j}^l\notag\\
&=\frac1{N(N-1)}\sum_{j=2}^k\sum_{i_1=1}^N\sum_{i_j\ne i_1}
 \mu_{i_1,i_j}^l\notag\\
&=(k-1)\avg \mu_{ij}^l,
\end{align}
where the expectation on the last line (and below in the proof)
is computed with respect to a pair of uniformly chosen distinct random vectors from $\cC$.
Next using \eqref{eq:ppm} and switching to the variable $w=(m/2)(1-\mu)$, we obtain
   \begin{align*}
\avg \mu_{ij}^l& = \Big(\frac2m\Big)^l \sum_{w=1}^{m}\frac{B_w}{N-1}
\Big(w-\frac{m}2\Big)^l\\
&= \Big(\frac2m\Big)^l \frac{N}{N-1}\Big[\sum_{w=0}^{m}\frac{B_w}{N} \Big(w-\frac{m}2\Big)^l -
\frac1N\Big(\frac{m}{2}\Big)^l \Big]\\
&= \Big(\frac2m\Big)^l \frac{N}{N-1}\Big[\frac1{2^m}\sum_{w=0}^{m}\binom{m}{w} \Big(w-\frac{m}2\Big)^l -
\frac1N\Big(\frac{m}{2}\Big)^l \Big],
\end{align*}
Now we can use \eqref{eq:ecm} and $l<m$ to write
  $$
  \avg\mu_{ij}^l \le\Big(\frac l{em}\Big)^{l/2}\frac N{N-1}
  \sqrt{l e^{1/3}}
    -\frac 1{N-1}\le e^{1/6} l^{(l+1)/2}(em)^{-l/2}.
   $$
Conclude using the condition on $m:$
  $$
  P_{R_k}\Big(\exists i\in I: \sum_{j\in I_i}\mu_{ij}> \delta\Big)\le
   k^{l+1} \delta^{-l}e^{1/6} l^{(l+1)/2}(em)^{-l/2} < \epsilon.
  $$
\vspace*{-.2in}\hfill\end{proof}

Observe that the condition of this theorem is nonasymptotic, and is
satisfied by a number of known constructions of orthogonal arrays.

\vspace*{.05in}{\em Example:} {\em
Consider sampling matrices obtained from the binary Delsarte-Goethals codes already mentioned above;
see Eq.\eqref{eq:example}.
It is known that the underlying code forms an orthogonal array of strength $t=7,$
so taking $l=6$ we obtain a family of $(k,\delta,\epsilon)$-StRIP matrices of dimensions
$m\times N$ for sparsity
   $$
  k\le 0.52\,(\delta^6\epsilon\, m^3)^{1/7}=0.52(\delta^6\epsilon)^{1/7}(2^rN)^{3/(7(r+2))}.
   $$
The case $r=0$ was considered in \cite{cal10b} where these matrices were analyzed
based on the detailed properties of this particular case of the construction.
Our computation, while somewhat crude, permits a uniform estimate for the entire family
of matrices.
\remove{
For instance, let $\cC$ be a binary
 Delsarte-Goethals code with $m=2^{2s+2},N=\frac12m^3, \mu=2m^{-\half},$ where $s\ge 0$ \cite[p.461]{mac91}.  It
  is known that $\cC$ forms an orthogonal array of strength $t=7.$ Taking $l=6$ we
  obtain a $(k,\delta,\epsilon)$-StRIP matrix of dimensions $m\times N$ where $k\le
  0.52 \delta^{6/7}\epsilon^{1/7} m^{3/7}.$}
The estimate can be improved if the expectation
$\avg\mu_{ij}^l$ can be computed explicitly from the known distribution of
correlations. For instance, taking $r=1$ and using the distribution given in \cite[p.477]{mac91} we obtain that $\avg{\mu^6}\approx(4/3)m^{-3}.$
\remove{Then the condition on $k$ and $m$ is given by the inequality
  $$
   \epsilon>\frac 43\frac{ k^7}{\delta^6}  m^{-3}.
  $$
The}
With this, the condition on sparsity that emerges has the form
$k<0.95 (\delta^{6}\epsilon m^{3})^{1/7},$ with a better constant compared to the general
estimate. For instance, we obtain $m\times (m^3/2)$ matrices with the $(k,\delta,0.001)$ StRIP
property for all $k\le 0.35  \delta^{6/7}m^{3/7}.$

Another similar possibility arises if $\cC$ is taken to be a binary dual BCH
code with $m=2^s-1, N=m^r, \mu=2(r-1)m^{-\half}, r=1,2,3,\dots.$ Many more such constructions
can be obtained
from other algebraic codes such as the Kerdock codes, Gold codes, etc. \cite{hel98}.
This lends further support to earlier studies of sampling matrices
constructed from the BCH codes \cite{ail09}, Delsarte-Goethals codes, and other binary
codes related to the second-order Reed-Muller codes \cite{cal10a,cal10b}.}

\vspace*{.05in}
It would be desirable to show that orthogonal arrays also suffice for the SINC property;
however, the technique introduced above results in parameters that contradict the Rao bound
on the number of rows in an array \cite{hed99}. Thus, we are unable to show that this construction
results in matrices that are good for linear estimators.

\subsection{Further constructions from binary codes} We remark that it is easy to show existence of matrices with low coherence.
The following observation is a rephrasing of the result known in coding theory 
as the Gilbert-Varshamov existence bound for binary linear codes. 
\begin{proposition} Let $l=\log_2 N, l<m$ and let $G=(\bfg_1,\dots,\bfg_l)$
be an $m\times l$ binary  matrix  whose rows are chosen independently and
uniformly from $\ff_2^l.$ Let $m=4\log N/\mu^2$, where $0<\mu<1.$
Form the matrix $\Phi$ by constructing an $\ff_2$-linear span of the columns of $G$
and using the map $\{0,1\}\to \{\frac 1{\sqrt{m}},\frac {-1}{\sqrt{m}}\}.$
Then $\Phi$ has coherence $\mu$ with
probability at least $1-2/N$ and mean square coherence $\bar\mu^2<1/m$ with 
probability at least $(1-(m/N))^m.$
\end{proposition}
\begin{proof}
Note that the Hamming distance $d$ between any
two columns of a matrix with coherence $\mu$ satisfies $\mu\ge|1-2d/m|.$
The set of columns of $C$ forms a linear space, so it suffices to argue
about Hamming weights rather than pairwise correlations.
Let $\bfu\in\{0,1\}^l$ be a nonzero vector, then
the probability that the vector $\bfv=G \bfu$ has weight $w$ equals
$\binom mw2^{-m}.$ Let $X$ be the random number of columns with weight
$|w-m/2|\ge m\mu/2.$ We have
  \begin{equation}\label{eq:cor}
\avg X \le 2\frac{N-1}{2^m}\sum_{w= 0}^{m(\frac{1}{2}-\frac{\mu}{2})} \binom mw
\le N 2^{1-m(1-h(\frac{1}{2}-\frac{\mu}{2}))}
  \end{equation}
where $h(x)=-x \log_2 x-(1-x) \log_2(1-x)$ is the binary entropy function.
Using the inequality
  $$
   1-h(\half-x)\ge 2x^2/\log 2, \quad 0\le x<1/2
  $$
and the condition for $\mu,$ we obtain $\avg X\le 2/N.$ Since $P(X>0)\le \avg X,$
this implies the first claim. The second part follows because there are $\prod_{i=1}^{m}(N-i)$
matrices $G$ with distinct nonzero rows.
\end{proof}
The derandomizing of Gilbert-Varshamov
codes was recently addressed by Porat and Rothschild \cite{por08}.
They presented a $O(mN)$ deterministic algorithm that constructs codes
with large minimum distance. To construct incoherent dictionaries, we need a bit more, 
namely that all the pairwise distances are in a narrow segment around $m/2.$ The
algorithm in \cite{por08} can be easily tailored to do this. 
A simplified version of this procedure which results in
the algorithm of complexity $O(mN^2)$ (i.e., not as good as in \cite{por08}), was given in \cite{maz11}.
In a nutshell it is as follows.
Instead of constructing the $m\times N$ matrix, $N=2^l$, we aim at constructing a
basis of the space of columns, i.e., an $m\times l$ matrix $G$. The rows of $G$
are selected recursively. Before any rows are selected, the expected number of
codewords of weight far from $m/2$ is given by \eqref{eq:cor}. The algorithm
selects rows one by one so that the expectation of the number of outlying
vectors {\em conditional} on the rows already chosen is the smallest possible.

We note that in the context of sparse recovery, the dependence between $N$ and
$m$ is likely to be polynomial. In this range of parameters the above
complexity is acceptable and is in fact comparable with the size of the matrix
$\Phi$ which needs to be stored for sampling and processing.

\providecommand{\href}[2]{#2}

\end{document}